\documentclass{irmaart}
\usepackage[pagebackref,final]{hyperref}
\usepackage{amsrefs,lineno,ha,mathrsfs,pgf,tikz,bbold,bm}
\usetikzlibrary{arrows,automata,decorations,shapes}
\usepackage{people}

\newif\ifnazi
\nazifalse

\pgfdeclaredecoration{Fractal Tree}{init}
{
  \state{init}[width=\pgfdecoratedinputsegmentremainingdistance]
  {
    \pgfpathlineto{\pgfpoint{\pgfdecoratedinputsegmentremainingdistance}{0pt}}
    \pgfpathmoveto{\pgfpoint{.55\pgfdecoratedinputsegmentremainingdistance}{0pt}}
    \pgfpathlineto{\pgfpoint{.55\pgfdecoratedinputsegmentremainingdistance}{.45\pgfdecoratedinputsegmentremainingdistance}}
    \pgfpathmoveto{\pgfpoint{.55\pgfdecoratedinputsegmentremainingdistance}{0pt}}
    \pgfpathlineto{\pgfpoint{.55\pgfdecoratedinputsegmentremainingdistance}{-.45\pgfdecoratedinputsegmentremainingdistance}}
    \pgfpathmoveto{\pgfpoint{.55\pgfdecoratedinputsegmentremainingdistance}{0pt}}
    \pgfpathlineto{\pgfpoint{\pgfdecoratedinputsegmentremainingdistance}{0pt}}
  }
}

\newcommand{\Card}[1]{\ifnazi\operatorname{Card}#1\else\##1\fi}
\newcommand{\At}{{\tilde A}}
\newcommand{\Ah}{{\hat A}}
\newcommand{\llangle}{\langle\!\langle}
\newcommand{\rrangle}{\rangle\!\rangle}
\newenvironment{fsa}[1][auto]{\begin{tikzpicture}[->,>=stealth',
    shorten >=1pt,auto,node distance=3cm,double distance between line centers=0.45ex,
    initial text=,accepting/.style=accepting by arrow,
    every loop/.style={looseness=12},semithick,#1]}{\end{tikzpicture}}
\newcommand\one{{\mathbb 1}}
\newcommand\mealy{{\mathcal M}}
\newcommand\End{{\operatorname{End}}}
\newcommand\Isom{{\operatorname{Isom}}}
\newcommand\GL{{\operatorname{GL}}}
\newcommand\chapterref[1]{24.\ref{#1}}

\title{Groups defined by automata}

\author{Laurent Bartholdi$^{1}$ \and Pedro V.~Silva$^{2,}$\thanks{The
    second author acknowledges support by Project ASA
    (PTDC/MAT/65481/2006) and C.M.U.P., financed by F.C.T. (Portugal)
    through the programmes POCTI and POSI, with national and
    E.U. structural funds.}}

\date{2010-12-07}

\address{$^1$~Mathematisches Institut\\
  Georg-August Universit\"at zu G\"ottingen\\
  Bunsenstra\ss e 3--5\\
  D-37073 G\"ottingen, Germany\\
  email:\,\url{laurent.bartholdi@gmail.com}
  \\[4pt]
  $^2$~Centro de Matem\'{a}tica, Faculdade de Ci\^{e}ncias\\
  Universidade do Porto\\
  R. Campo Alegre 687\\
  4169-007 Porto, Portugal\\
  email:\,\url{pvsilva@fc.up.pt}}

\begin{document}
\maketitle\label{chapterBS2}

\vspace{-1cm}



\vspace{-5mm}

\localtableofcontents

Finite automata have been used effectively in recent years to define
infinite groups. The two main lines of research have as their most
representative objects the class of automatic groups (including
``word-hyperbolic groups'' as a particular case) and automata groups
(singled out among the more general ``self-similar groups'').

The first approach is studied in Section~\ref{BS2:sec:1} and implements
in the language of automata some tight constraints on the geometry of
the group's Cayley graph.  Automata are used to define a normal form
for group elements and to execute the fundamental group operations.

The second approach is developed in Section~\ref{BS2:sec:2} and focuses
on groups acting in a finitely constrained manner on a regular rooted
tree. The automata define sequential permutations of the tree, and can
even represent the group elements themselves.

The authors are grateful to Martin R. Bridson, Fran\c cois Dahmani,
Rostislav I. Grigorchuk, Luc Guyot, and Mark V. Sapir for their
remarks on a preliminary version of this text.

\section{The geometry of the Cayley graph}\label{BS2:sec:1}
Since its inception at the beginning of the 19th century, group theory
has been recognized as a powerful language to capture
\emph{symmetries} of mathematical objects: crystals in the early 19th
century, for Hessel and Frankenheim~\cite{Engel}*{page~120}; roots of
a polynomial, for Galois and Abel; solutions of a differential
equation, for Lie, Painlev\'e, etc. It was only later, mainly through
the work of Klein and Poincar\'e, that the tight connections between
group theory and geometry were brought to light.

Topology and group theory are related as follows.  Consider a space
$X$, on which a group $G$ acts \emph{freely}\index{action!free}: for
every $g\neq\one\in G$ and $x\in X$, we have $x\cdot g\neq x$. If the
quotient space $Z=X/G$ is compact, then $G$ ``looks very much like''
$X$, in the following sense: choose any $x\in X$, and consider the
orbit $x\cdot G$. This identifies $G$ with a roughly evenly
distributed subset of $X$.

Conversely, consider a ``nice'' compact space $Z$ with
\emph{fundamental group}\index{fundamental~group} $G$: then
$X=\widetilde Z$, the \emph{universal cover}\index{universal~cover} of
$Z$, admits a free $G$-action. In conclusion, properties of the
fundamental group of a compact space $Z$ reflect geometric properties
of the space's universal cover.

We recall that finitely generated groups were defined
in~\S\ref{BS1:fg}: they are groups $G$ admitting a
surjective map $\pi:F_A\twoheadrightarrow G$, where $F_A$ is the free
group on a finite set $A$.

\begin{definition}
  A group $G$ is \emph{finitely
    presented}\index{group!finitely~presented} if it is finitely
  generated, say by $\pi:F_A\twoheadrightarrow G$, and if there exists a
  finite subset $\mathscr R\subset F_A$ such the kernel $\ker(\pi)$ is
  generated by the $F_A$-conjugates of $\mathscr R$, that is,
  $\ker(\pi)=\llangle\mathscr R\rrangle$; one then has
  $G=F_A/\llangle\mathscr R\rrangle$. These $r\in \mathscr R$ are called
  \emph{relators}\index{group!relators~of~a} of the presentation; and one writes
  \[G=\langle A\mid \mathscr R\rangle.\]
  Sometimes it is convenient to write a relator in the form `$a=b$'
  rather than the more exact form `$ab^{-1}$'.
\end{definition}

Let $G$ be a finitely generated group, with generating set $A$. Its
\emph{Cayley graph}\index{Cayley~graph}\index{graph!Cayley} $\mathscr
C(G,A)$, introduced by \Cayley\ \cite{Cayley}, is the graph with
vertex set $G$ and edge set $G\times A$; the edge $(g,s)$ starts at
vertex $g$ and ends at vertex $gs$.

In particular, the group $G$ acts freely on $\mathscr C(G,A)$ by left
translation; the quotient $\mathscr C(G,A)/G$ is a graph with one
vertex and $\Card A$ loops.

Assume moreover that $G$ is finitely presented, with relator set
$\mathscr R$. For each $r=r_1\cdots r_n\in\mathscr R$ and each $g\in
G$, the word $r$ traces a closed path in $\mathscr C(G,A)$, starting
at $g$ and passing successively through
$gr_1,gr_1r_2,\dots,gr_1r_2\cdots r_n=g$. If one ``glues'' for each
such $r,g$ a closed disk to $\mathscr C(G,A)$ by identifying the disk's
boundary with that path, one obtains a $2$-dimensional cell complex in
which each loop is contractible --- this is a direct translation of
the fact that the normal closure of $\mathscr R$ is the kernel of the
presentation homomorphism $F_A\to G$.

For example, consider $G=\Z^2$, with generating set
$A=\{(0,1),(1,0)\}$. Its Cayley graph is the standard square grid. The
Cayley graph of a free group $F_A$, generated by $A$, is a tree.
\[\begin{tikzpicture}
   \draw[step=.5cm,thin] (-0.2,-0.2) grid (4.7,2.7);
 \end{tikzpicture}
 \qquad
 \begin{tikzpicture}[decoration=Fractal Tree,thin]
   \clip (-2.7,-1.7) rectangle (2.7,1.7);
   \draw decorate{ decorate{ decorate { decorate { (0,0) -- (2.5,0) } } } };
   \draw decorate{ decorate{ decorate{ decorate { (0,0) -- (-2.5,0) } } } };
   \draw decorate{ decorate{ decorate{ decorate { (0,0) -- (0,2.5) } } } };
   \draw decorate{ decorate{ decorate{ decorate { (0,0) -- (0,-2.5) } } } };
 \end{tikzpicture}\]

 More generally, consider a right $G$-set $X$, for instance the coset
 space $H\backslash G$. The \emph{Schreier
   graph}\index{Schreier~graph}\index{graph!Schreier} $\mathscr
 C(G,X,A)$ of $X$ is then the graph with vertex set $X$ and edge set
 $X\times A$; the edge $(x,s)$ starts in $x$ and ends in $xs$.

\subsection{History of geometric group theory}
In a remarkable series of papers,
\Dehn~\cites{MR1511580,MR1511645,MR1511705}, see also~\cite{MR881797},
initiated the geometric study of infinite groups, by trying to relate
algorithmic questions on a group $G$ and geometric questions on its
Cayley graph. These problems\index{problem,~decision} were described
in Definition~\ref{BS1:decproblems}, to which we refer. For instance,
the word problem asks if one can determine whether a path in the
Cayley graph of $G$ is closed, knowing only the path's labels.

It is striking that \Dehn\ used, for Cayley graph, the German
\emph{Gruppenbild}, literally ``group picture''. We must solve the
word problem in a group $G$ to be able to draw bounded portions of its
Cayley graph; and some algebraic properties of $G$ are tightly bound
to the algorithmic complexity of the word problem,
see~\S\ref{BS1:beyondfg}. For example, \Muller\ and \Schupp\ prove (see
Theorem~\ref{BS1:musc}) that a push-down automaton recognizes
precisely the trivial elements of $G$ if and only if $G$ admits a free
subgroup of finite index.

We consider now a more complicated example. Let $\mathcal S_g$ be an
oriented surface of genus $g\ge2$, and let $J_g$ denote its fundamental
group\index{surface~group}\index{group!surface}. Recall that $[x,y]$
denotes in a group the commutator $x^{-1}y^{-1}xy$. We have a presentation
\begin{equation}\label{BS2:eq:jg}
  J_g=\langle a_1,b_1,\dots,a_g,b_g\mid[a_1,b_1]\cdots[a_g,b_g]\rangle.
\end{equation}
Let $r=[a_1,b_1]\cdots[a_g,b_g]$ denote the relator, and let $\mathscr
R^*$ denote the set of cyclic permutations of $r^{\pm1}$.  The word
problem in $J_g$ is solvable in polynomial time by the following
algorithm: let $u$ be a given word. Freely reduce $u$ by removing all
$aa^{-1}$ subwords. Then, if $u$ contains a subword $v_1$ such that
$v_1v_2\in\mathscr R^*$ and $v_1$ is longer than $v_2$, replace
$v_1$ by $v_2^{-1}$ in $u$ and repeat. Eventually, $u$ represents
$\one\in G$ if and only if it is the empty word.\index{word~problem}

The validity of this algorithm relies on a lemma by \Dehn, that every
nontrivial word representing the identity contains more than half of
the relator as a subword.

Incidentally, the Cayley graph of $J_g$ is a tiling of the hyperbolic
plane\index{hyperbolic~plane} by $4g$-gons, with $4g$ meeting at each vertex.

\Tartakovsky~\cite{MR0033816},
\Greendlinger~\cites{MR0125020,MR0124381} and
\Lyndon~\cites{MR0577064,MR0214650} then devised ``small
cancellation''\index{small~cancellation} conditions on a group
presentation that guarantee that \Dehn's algorithm will
succeed. Briefly said, they require the relators to have small enough
overlaps. These conditions are purely combinatorial, and are described
in~\S\chapterref{BS2:sec:exautom}.

\Cannon\ and \Thurston, on the other hand, sought a formalism that
would encode the ``periodicity of pictures'' of a group's Cayley
graph. Treating the graph as a metric space with geodesic distance
$d$, already seen in~\S\ref{BS1:conjugacy}, they make the following
definition: the \emph{cone type}\index{cone~type} of $g\in G$ is
\begin{equation}\label{BS2:eq:cone}
  C_g=\{h\in G\mid d(\one,gh)=d(\one,g)+d(g,gh)\};
\end{equation}
the translate $gC_g$ is the set of vertices that may be connected to
$\one$ by a geodesic passing through $g$. Their intuition is that the
cone type of a vertex $v$ remembers, for points near $v$, whether they
are closer or further to the origin than $v$; for example, $\Z^2$ with
its standard generators has $9$ cone types: the cone type of the
origin (the whole plane), those of vertices on the axes (half-planes),
and those of other vertices (quadrants).

\Thurston's motivation was to get a good, algorithmic understanding of
fundamental groups\index{fundamental~group} of threefolds.
They should be made of nilpotent (or, more generally, solvable)
groups on the one hand, and ``automatic'' groups on the other hand.

\begin{definition}
  Let $G=\langle A\rangle$ be a finitely generated group, and recall
  that $\At$ denotes $A\sqcup A^{-1}$. The \emph{word
    metric}\index{word~metric}\index{distance!word~metric} on $G$ is
  the geodesic distance in $G$'s Cayley graph $\mathscr C(G,A)$. It
  may be defined directly as
  \[d(g,h)=\min\{n\mid g=hs_1\cdots s_n\text{ with all
  }s_i\in\At\},\] and is left-invariant: $d(xg,xh)=d(g,h)$.
  The \emph{ball of radius $n$}\index{group!ball~of~radius~$n$~in} is the set
  \[B_{G,A}(n)=\{g\in G\mid d(\one,g)\le n\}.\] The \emph{growth
    function}\index{growth~function} of $G$ is the function
  \[\gamma_{G,A}(n)=\Card B_{G,A}(n).\]
  The \emph{growth series}\index{growth~series} of $G$ is the power series
  \[\Gamma_{G,A}(z)=\sum_{g\in G}z^{d(\one,g)}=\sum_{n\ge0}\gamma_{G,A}(n)z^n(1-z).\]

  Growth functions are usually compared as follows:
  $\gamma\precsim\delta$ if there is a constant $C\in\N$ such that
  $\gamma(n)\le\delta(Cn)$ for all $n\in\N$; and $\gamma\sim\delta$ if
  $\gamma\precsim\delta\precsim\gamma$. The equivalence class of
  $\gamma_{G,A}$ is independent of $A$.
\end{definition}

\Cannon\ observed (in an unpublished 1981 manuscript; see
also~\cite{MR758901}) that, if a group has finitely many cone types,
then its growth series satisfies a finite linear system and is
therefore a rational function of $z$. For $J_g$, for instance, he
computes
\[\Gamma_{J_g,A}=\frac{1+2z+\dots+2z^{2g-1}+z^{2g}}{1+(2-4g)z+\dots+(2-4g)z^{2g-1}+z^{2g}}.\]\index{surface~group!growth~series}
This notion was formalized by \Thurston\ in 1984 using
automata, and is largely the topic of the next section. We will return
to growth of groups in~\S\chapterref{BS2:sec:growth}; see
however~\cite{MR1344918} for a good example of growth series of groups
computed thanks to a description of the Cayley graph by automata.

\Gromov\ emphasized the relevance to group theory of the following
definition, attributed to \Margulis:
\begin{definition}[\cite{MR804694}]\label{BS2:def:geom}
  A map $f:X\to Y$ between two metric spaces is a
  \emph{$C$-quasi-isometry}\index{quasi-isometry}, for a constant
  $C>0$, if one has
  \[C^{-1}d(x,y)-C\le d(f(x),f(y))\le Cd(x,y)+C
  \]
  for all $y\in Y$ such that $d(f(X),y)\le C$. A \emph{quasi-isometry}
  is a $C$-quasi-isometry for some $C>0$.  Two spaces are
  \emph{quasi-isometric} if there exists a quasi-isometry between
  them; this is an equivalence relation.

  A property of finitely generated groups is
  \emph{geometric}\index{property!geometric} if it
  only depends on the quasi-isometry class of its Cayley graph.
\end{definition}
Thus for instance the inclusion $\Z\to\R$, and the map
$\R\to\Z,x\mapsto\lfloor x\rfloor$ are quasi-isometries.

Being finite, having a finite-index subgroup isomorphic to $\Z$, and
being finitely presented are geometric properties. The asymptotics of
the growth function is also a geometric invariant; thus for instance
having growth function $\precsim n^2$ is a geometric property.

\subsection{Automatic groups}\index{automatic~group|(}
Let $G=\langle A\rangle$ be a finitely generated group. We will
consider the formal alphabet $\Ah=A\sqcup
A^{-1}\sqcup\{\one\}$, where $\one$ is treated as a ``padding''
symbol. Following the main reference~\cite{MR1161694} by \Epstein\
\emph{et al.}:

\begin{definition}[\cites{MR1161694,MR1189009,MR1147304}]\label{BS2:def:automatic}
  The group $G$ is
  \emph{automatic}\index{automatic~group}\index{group!automatic} if
  there are finite-state automata
  $\mathcal L,\mathcal M$, the \emph{language} and
  \emph{multiplication} automata, with the following properties:
  \begin{conditionsiii}
  \item $\mathcal L$ is an automaton with alphabet $\At$;
  \item $\mathcal M$ has alphabet $\Ah\times\Ah$, and
    has for each $s\in \Ah$ an accepting subset $T_s$ of states;
    call $\mathcal M_s$ the automaton with accepting states $T_s$;
  \item the language of $\mathcal L$ surjects onto $G$ by the natural map
    $f:\At\to F_A\to G$; words in $L(\mathcal L)$ are called
    \emph{normal forms}\index{automatic~group!normal~forms};
  \item for any two normal forms $u,v\in L(\mathcal L)$, consider the word
    \[w=(u_1,v_1)(u_2,v_2)\cdots(u_n,v_n)\in(\Ah\times\Ah)^*,\] where
    $n=\max\{|u|,|v|\}$ and $u_i,v_j=\one$ if $i>|u|,j>|v|$. Then
    $\mathcal M_s$ accepts $w$ if and only if $\pi(u)=\pi(vs)$.
  \end{conditionsiii}
\end{definition}

In words, $G$ is automatic if the automaton $\mathcal L$ singles out
sufficiently many words which may be used to represent all group
elements; and the automaton $\mathcal M_s$ recognizes when two such
singled out words represent group elements differing by a
generator. The pair $(\mathcal L,\mathcal M)$ is an \emph{automatic
  structure}\index{automatic~structure} for $G$.

We will give numerous examples of automatic groups
in~\S\chapterref{BS2:sec:exautom}. Here is a simple one that contains the main
features: the group
$G=\Z^2$\index{abelian~group!free}\index{group!free~abelian}, with
standard generators $x,y$. The
language accepted by $\mathcal L$ is
$(x^*\cup(x^{-1})^*)(y^*\cup(y^{-1})^*)$:
\[\begin{fsa}
  \node[state] (1) at (0,0) {};
  \node[state,accepting] (x) at (3,0) {};
  \node[state,accepting left] (X) at (-3,0) {};
  \node[state,accepting] (y) at (0,2) {};
  \node[state,accepting] (Y) at (0,-2) {};
  \path (135:1) edge (1);
  \path (1) edge (315:1);
  \path (1) edge node {$x$} (x) edge node {$x^{-1}$} (X)
            edge node {$y$} (y) edge node {$y^{-1}$} (Y);
  \path (x) edge[loop above] node {$x$} () edge node {$y$} (y) edge node {$y^{-1}$} (Y);
  \path (X) edge[loop above] node {$x^{-1}$} () edge node {$y$} (y) edge node {$y^{-1}$} (Y);
  \path (y) edge[loop left] node {$y$} ();
  \path (Y) edge[loop left] node {$y^{-1}$} ();
\end{fsa}\]

The multiplication automaton, in which states in $T_s$ are labeled $s$, is
\[\begin{fsa}
  \node[state] (1) at (0,0) {$\one$};
  \node[state] (x) at (3,0) {$x$};
  \node[state] (xy) at (3,3) {};
  \node[state] (xY) at (3,-3) {};
  \node[state] (y) at (0,3) {$y$};
  \node[state] (Y) at (0,-3) {$y^{-1}$};
  \node[state] (X) at (-3,0) {$x^{-1}$};
  \node[state] (Xy) at (-3,3) {};
  \node[state] (XY) at (-3,-3) {};
  \path (202:1) edge (1);
  \small
  \path (1) edge[loop above,left=3mm] node {$(s,s)$} ()
                    edge node[sloped,above] {$(x,y^{-1})$} node[sloped,below] {$(y,x^{-1})$} (xy)
                    edge node[sloped,above] {$(x,y)$} node[sloped,below,near end] {$(y^{-1},x^{-1})$} (xY)
                    edge node[near end,right] {$(y,\one)$} node[near end] {$(\one,y^{-1})$} (y)
                    edge node {$(x,\one)$} node[below] {$(\one,x^{-1})$} (x)
                    edge node[sloped,below] {$(x^{-1},y^{-1})$} node[sloped,above,near end] {$(y,x)$} (Xy)
                    edge node[sloped,below] {$(x^{-1},y)$} node[sloped,above,near end] {$(y^{-1},x)$} (XY)
                    edge node {$(\one,x)$} node[above] {$(x^{-1},\one)$} (X)
                    edge[bend left=25] node[near end] {$(\one,y)$} node[left,near end] {$(y^{-1},\one)$} (Y);
  \path (xy) edge[loop right] node {$(s,s)$} ()
             edge node[near start] {$(y^{-1},\one)$} node [near end] {$(\one,y)$} (x);
  \path (xY) edge[loop right] node {$(s,s)$} ()
             edge node [near start,right] {$(y,\one)$} node [near end,right] {$(\one,y^{-1})$} (x);
  \path (Xy) edge[loop left] node {$(s,s)$} ()
             edge node[near start,left] {$(y^{-1},\one)$} node [near end,left] {$(\one,y)$} (X);
  \path (XY) edge[loop left] node {$(s,s)$} ()
             edge node[near start] {$(y,\one)$} node [near end] {$(\one,y^{-1})$} (X);
\end{fsa}\]

The definition we gave is purely automata-theoretic. It does, however,
have a more geometric counterpart. A word $w\in\At^*$
represents in a natural way a path in the Cayley graph $\mathscr
C(G,A)$, starting at $\one$ and ending at $\pi(w)$. If $w=w_1\cdots
w_n$, we write $w(j)=w_1\cdots w_j$ the vertex of $\mathscr C(G,A)$
reached after $j$ steps; if $j>n$ then $w(j)=w$. For two paths
$u,v\in\At^*$, we say they
\emph{$k$-fellow-travel}\index{property!fellow~traveller}\index{fellow~traveller~property}
if
$d(u(j),v(j))\le k$ for all $j\in\{1,\dots,\max\{|u|,|v|\}\}$.
\begin{proposition}\label{BS2:prop:autom}
  A group $G$ is automatic if and only if there exists a rational
  language $L\subseteq\At^*$, mapping onto $G$, and a
  constant $k$, such that for any $u,v\in L$ with $d(\pi(u),\pi(v))\le
  1$ the paths $u,v$ $k$-fellow-travel.
\end{proposition}
\begin{proof}[Sketch of proof]
  Assume first that $G$ has automatic structure $(\mathcal L,\mathcal
  M)$, and let $c$ denote the number of states of $\mathcal M$. If
  $u,v\in L(\mathcal L)$ satisfy $\pi(u)=\pi(vs)$, let $s_j$
  denote the state $\mathcal M$ is in after having read
  $(u_1,v_1)\cdots(u_j,v_j)$. There is a path of length $<c$, in
  $\mathcal M$, from $s_j$ to an accepting state (labeled $s$); let
  its label be $(p,q)$. Then $\pi(u(j)p)=\pi(v(j)qs)$, so
  $u(j)$ and $v(j)$ are at distance at most $2c-1$ in $\mathscr
  C(G,A)$.

  Conversely, assume that paths $k$-fellow-travel and that an
  automaton $\mathcal L$, with state set $Q$ is given, with language
  surjecting onto $G$. Recall that $B(k)$ denotes the set of group
  elements at distance $\le k$ from $\one$ in $\mathscr
  C(G,A)$. Consider the automaton with state set $Q\times Q\times
  B_k$. Its initial state is $(*,*,\one)$, where $*$ is the initial
  state of $\mathcal L$; its alphabet is $\Ah\times\Ah$, and its
  transitions are given by $(p,q,g)\cdot(s,t)=(p\cdot s,q\cdot
  t,s^{-1}gt)$ whenever these are defined. Its accepting set of
  states, for $s\in\Ah$, is $T_s=Q\times Q\times\{s\}$.
\end{proof}

\begin{corollary}
  If the finitely generated group $G=\langle A\rangle$ is automatic,
  and if $B$ is another finite generating set for $G$, then there also
  exists an automatic structure for $G$ using the alphabet $B$.
\end{corollary}
\begin{proof}[Sketch of proof]
  Note first that a trivial generator may be added or removed from $A$
  or $B$, using an appropriate finite transducer for the latter.

  There exists then $M\in\N$ such that every $a\in\At$ can be
  written as a word $w_a\in\tilde B^*$ of length precisely
  $M$. Accept as normal forms all $w_{a_1}\cdots w_{a_n}$ such that
  $a_1\cdots a_n$ is a normal form in the original automatic structure
  $\mathcal L$. The new normal forms constitute a homomorphic image of
  $\mathcal L$ and therefore define a rational language. If paths in
  $L(\mathcal L)$ $k$-fellow-travel, then their images in the new
  structure will $kM$-fellow-travel.
\end{proof}

Note that the language of normal forms is only required to contain
``enough'' expressions; namely that the evaluation map $L(\mathcal
L)\to G$ is onto. We may assume that it is bijective, by the following
lemma. The language $L(\mathcal L)$ is then called a ``rational
cross-section''\index{rational~cross-section} by
\Gilman~\cite{MR895616}; and $(\mathcal L,\mathcal M)$ is called an
\emph{automatic structure with
  uniqueness}\index{automatic~structure!with~uniqueness}.

\begin{lemma}\label{BS2:lem:uniquerat}
  Let $G$ be an automatic group. Then $G$ admits an automatic
  structure with uniqueness.
\end{lemma}
\begin{proof}[Sketch of proof]
  Consider $(\mathcal L',\mathcal M)$ an automatic structure. Recall the
  ``short-lex'' ordering on words: $u\le v$ if $|u|<|v|$, or if
  $|u|=|v|$ and $u$ comes lexicographically before $v$. The language
  $\{(u,v)\in\Ah^*\times\Ah^* \mid u\le v\}$ is
  rational. The language
  \[L=L(\mathcal L')\cap\{u\in\At^*\mid
\mbox{ for all }v\in\Ah^*,\mbox{ if }(u,v)\in L(\mathcal M_\one)\mbox{
  then }u\le v\}\] is then
  also rational, of the form $L(\mathcal L)$. The automaton $\mathcal
  M$ need not be changed.
\end{proof}

\noindent Various notions related to automaticity have emerged, some stronger,
some weaker:
\begin{itemize}
\item One may require the words accepted by $\mathcal L$ to be
  representatives of minimal length; the automatic structure is then
  called \emph{geodesic}\index{automatic~structure!geodesic}. It would
  then follow that the growth series
  $\Gamma_{G,A}(z)$ of $G$, which is the growth series of $\mathcal
  L$, is a rational function. Note that there is a constant $K$ such
  that, for the language produced by Lemma~\ref{BS2:lem:uniquerat},
  all words $u\in L(\mathcal L)$ satisfy $|u|\le Kd(\one,\pi(u))$.
\item The definition is asymmetric; more precisely, we have defined a
  \emph{right automatic} group, in that the automaton $\mathcal M$
  recognizes multiplication on the right. One could similarly define
  \emph{left automatic groups}\index{automatic~group!right/left}; then
  a group is right automatic if and
  only if it is left automatic.

  Indeed, let $(\mathcal L,\mathcal M)$ be an automatic structure where
  $\mathcal L$ recognizes a rational cross section. Then
  $L'=\{u^{-1}\mid u\in L(\mathcal L)\}$ and $M'=\{(u^{-1},v^{-1})\mid
  (u,v)\in L(\mathcal M)\}$ are again rational languages. Indeed,
  since rational languages are closed under reversal and morphisms,
  it follows easily that $L'$ is rational. On the other hand, using
  the pumping lemma and the fact that group elements admit unique
  representatives in $L(\mathcal L)$, the amount of padding at the end
  of word-pairs in $L(\mathcal M)$ is bounded, and can be moved from
  the beginning to the end of the word-pairs in $M'$ by a finite
  transducer. Therefore, $L',M'$ are the languages of a right automatic
  structure.

  However, one could require both properties simultaneously, namely,
  on top of an automatic structure, a third automaton $\mathcal N$
  accepting (in state $s\in\Ah$) all pairs of normal forms
  $(u,v)$ with $\pi(u)=\pi(sv)$. Such groups are called
  \emph{biautomatic}\index{biautomatic~group}\index{automatic~group!biautomatic}\index{group!biautomatic}. No
  example is known of a group that is automatic but not biautomatic.

\item One might also only keep the geometric notion of ``combing'': a
  \emph{combing}\index{combing~of~group} on a group is a choice, for every
  $g\in G$, of a word
  $w_g\in\At^*$ evaluating to $g$, such that the words $w_g$
  and $w_{gs}$ fellow-travel for all $g\in G$, $s\in\At$.

  In that sense, a group is automatic if and only if it admits a
  combing whose words form a rational language; see~\cite{MR2275631}
  for details.

  One may again require the combing lines to be geodesics, i.e., words
  of minimal length; see \Hermiller's
  work~\cites{MR1464930,MR2310150,MR2440717}.

  One may also put weaker constraints on the words of the combing; for
  example, require it to be an indexed language. \Bridson\ and
  \Gilman\ \cite{MR1420509} proved that all geometries of
  threefolds, in particular the Nil~\eqref{BS2:eq:heis} and Sol
  geometry, which are not automatic, fall in this framework.
\item Another relaxation is to allow the automaton $\mathcal M$ to
  read at will letters from the first or the second word; groups
  admitting such a structure are called \emph{asynchronously
    automatic}\index{group!asynchronously~automatic}. Among
  fundamental groups of threefolds, there is no
  difference between these definitions~\cite{MR1420509}, but for more
  general groups there is.
\item Finally, Definition~\ref{BS2:def:automatic} can be adapted to
  define automatic semigroups\index{semigroup!automatic}. Properties
  from automatic groups that
  can be proved within the automata-theoretic framework can often be
  generalized to automatic semigroups, or at least
  monoids~\cite{MR1795250}\index{monoid!automatic}. However,
  establishing an alternative
  geometric approach has proved to be a tough task and success was
  reached only in restricted cases~\cites{MR2082097,MR2277577}.
\end{itemize}

\subsection{Main examples of automatic groups}\label{BS2:sec:exautom}
From the very definition, it is clear that finite groups are
automatic: one chooses a word representing each group element, and
these necessarily form a fellow-travelling rational language.

It is also clear that $\Z$ is automatic: write $t$ for the canonical
generator of $\Z$; the language $t^*\cup(t^{-1})^*$ maps bijectively to
$\Z$; and the corresponding paths $1$-fellow-travel. The automata are
\[\mathcal L:\begin{fsa}[baseline]
  \node[state,initial above,accepting below] (1) at (0,0) {};
  \node[state,accepting below] (t) at (2,0) {};
  \node[state,accepting below] (T) at (-2,0) {};
  \path (1) edge node {$t$} (t) edge node {$t^{-1}$} (T);
  \path (t) edge[loop above] node {$t$} ();
  \path (T) edge[loop above] node {$t^{-1}$} ();
\end{fsa},\qquad\mathcal M:\begin{fsa}[baseline]
  \node[state,initial below] (1) at (0,0) {$\one$};
  \node[state] (t) at (2,0) {$t$};
  \node[state] (T) at (-2,0) {$t^{-1}$};
  \path (1) edge[loop above] node {$(s,s)$} ()
            edge node {$(t,\one)$} node[below] {$(\one,t^{-1})$} (t)
            edge node {$(t^{-1},\one)$} node[above] {$(\one,t)$} (T);
\end{fsa}.\]

Simple constructions show that the direct and free products of
automatic groups are again automatic. Finite extensions and
finite-index subgroups of automatic groups are automatic. It is
however still an open problem whether a direct factor of an automatic
group is automatic.

Recall that we glued disks, one for each $g\in G$ and each
$r\in\mathscr R$, to the Cayley graph of a finitely presented group
$G=\langle A\mid\mathscr R\rangle$, so as to obtain a $2$-complex
$\mathscr K$. The \emph{small cancellation
  conditions}\index{small~cancellation} express a
combinatorial form of non-positive curvature of $\mathscr K$: roughly,
$C(p)$ means that every proper edge cycle in $\mathscr K$ has length
$\ge p$, and $T(q)$ means that every proper edge cycle in the dual
$\mathscr K^\vee$ has length $\ge q$; see~\cite{MR0577064}*{Chapter~V}
for details. If $G$ satisfies $C(p)$ and $T(q)$ where
$p^{-1}+q^{-1}\le\frac12$, then $G$ is automatic.

Consider the configurations defined by $n$ strings in
$\R^2\times[0,1]$, with string 
$\#i$ starting at $(i,0,0)$ and ending at $(i,0,1)$; these
configurations are viewed up to isotopy preserving the
endpoints. They can be multiplied (by stacking them
above each other) and inverted (by flipping them up-down), yielding a
group, the \emph{pure braid group}; if the strings are allowed to end
in an arbitrary permutation, one obtains the \emph{braid
  group}\index{group!braid}. This
group $B_n$ is generated by elementary half-twists of strings
$\#i,i+1$ around each other, and admits the presentation
\[B_n=\langle\sigma_1,\dots,\sigma_{n-1}\mid\sigma_i\sigma_{i+1}\sigma_i=\sigma_{i+1}\sigma_i\sigma_{i+1},[\sigma_i,\sigma_j]\text{
  whenever }|i-j|\ge2\rangle.\] More generally, consider a surface
$\mathcal S$ of genus $g$, with $n$ punctures and $b$ boundary
components. The \emph{mapping class group}\index{group!mapping~class}
$M_{g,n,b}$ is the group of
maps $\mathcal S\to\mathcal S$ modulo isotopy, and $B_n$ is the special case
$M_{0,n,1}$ of mapping classes of the $n$-punctured disk. All mapping
class groups $M_{g,n,b}$ are automatic groups~\cite{MR1343324}.

As another generalization of braid groups, consider \emph{Artin
  groups}\index{group!Artin}\index{Artin~group}. 
Let $(m_{ij})$ be a
symmetric $n\times n$-matrix with entries in $\N\cup\{\infty\}$. 
The \emph{Artin group} of type
$(m_{ij})$ is the group with presentation
\[A(m)=\langle s_1,\dots,s_n\mid
(s_is_j)^{\lfloor m_{ij}/2\rfloor}=(s_js_i)^{\lfloor m_{ij}/2\rfloor}\text{ whenever 
}m_{ij}<\infty\rangle.\] 
The corresponding \emph{Coxeter
  group}\index{group!Coxeter}\index{Coxeter~group} has presentation
\[C(m)=\langle s_1,\dots,s_n\mid s_i^2,
(s_is_j)^{m_{ij}/2}=(s_js_i)^{m_{ij}/2}\text{ whenever }m_{ij}<\infty\rangle.\]
An Artin group $A(m)$ has \emph{finite type} if $C(m)$ is finite.
Artin groups of finite type are biautomatic~\cite{MR1157320}.  Coxeter
groups are automatic~\cite{MR1213378}.

Fundamental groups\index{fundamental~group!of~$3$-fold~is~automatic}
of threefolds, except those with a piece modelled on Nil or Sol
geometry~\cite{MR1161694}*{chapter 12}, are
automatic\Thurston[].

\subsection{Properties of automatic groups} The definition of
automatic groups, by automata, has a variety of interesting
consequences. First, automatic groups are finitely presented; more
generally, combable groups are finitely presented:
\begin{proposition}[\cite{MR1230633}]\label{BS2:prop:class}
  Let $G$ be a combable group. Then $G$ has type
  $F_\infty$\index{group!finfinity@$F_\infty$}\index{finfinity@$F_\infty$~group},
namely,
  there exists a contractible cellular complex with free $G$-action
  and finitely many $G$-orbits of cells in each dimension.
\end{proposition}
\noindent(Finite presentation is equivalent to ``finitely many $G$-orbits of
cells in dimension $\le2$'').
\begin{proof}[Sketch of proof]
  By assumption, $G$ is finitely generated. Therefore, the Cayley
  graph contains one $G$-orbit of $0$-cells (vertices), and $\Card A$
  orbits of $1$-cells (edges). Consider all pairs of paths $u,v$ in
  the combing that have neigbouring extremities. They
  $k$-fellow-travel by hypothesis; so there are for all $j$ paths
  $w(j)$ of length $\le k$ connecting $u(j)$ to $v(j)$. The closed
  paths $u(j)-v(j)-v(j+1)-u(j+1)-u(j)$ have length $\le2k+2$, so they trace
  finitely many words in $F_A$. Taking them as relators defines a
  finite presentation for $G$. The process may be continued with
  higher-dimensional cells.
\end{proof}

\begin{proposition}\label{BS2:prop:quadraticiso}
  Automatic groups satisfy a \emph{quadratic isoperimetric
    inequality}\index{isoperimetric~inequality}\index{automatic~group!quadratic~isoperimetric~inequality}; that is, for any finite
  presentation $G=\langle
  A\mid\mathscr R\rangle$ there is a constant $k$ such that, if
  $w\in F_A$ is a word evaluating to $\one$ in $G$, then
  \[w=\prod_{i=1}^\ell r_i^{w_i}\text{ for some }r_i\in\mathscr
  R^{\pm1},w_i\in F_A\text{ and }\ell\le k|w|^2.\]
\end{proposition}
\begin{proof}[Sketch of proof]
  Write $n=|w|$, and draw the combing lines between $\one$ and
  $w(j)$. There are $n$ combing lines, which have length $\mathcal
  O(n)$; so the gap between neighbouring combing lines can be filled
  by $\mathcal O(n)$ relators. This gives $\mathcal O(n^2)$ relators
  in total.
\end{proof}

Note that being finitely presented is usually of little value as far
as algorithmic questions are concerned: there are finitely presented
groups whose word problem cannot be solved by a Turing
machine~\cites{MR0075197,MR0101267}. By contrast:\index{word~problem}
\begin{proposition}\label{BS2:prop:wpautom}
  The word problem in a group given by an automatic structure is
  solvable in quadratic 
  time\index{automatic~group!word~problem~in~quadratic~time}. A word
  may even be put into canonical form in quadratic time.
\end{proposition}
\begin{proof}[Sketch of proof]
  We may assume, by Lemma~\ref{BS2:lem:uniquerat}, that every $g\in G$
  admits a unique normal form. Now, given a word $u=a_1\cdots
  a_n\in\Ah^*$, construct the following words: $w_0\in
  L(\mathcal L)$ is the representative of $\one$. Treating $\mathcal M_a$ as a
  non-deterministic automaton in its second variable, find for
  $i=1,\dots,n$ a word $w_i\in\Ah^*$ such that the padding of
  $(w_{i-1},w_i)$ is accepted by $\mathcal M_{a_i}$. Then
  $\pi(u)=\one\in G$ if and only if $w_n=w_0$.

  Clearly the $w_i$ have linear length in $i$, so the total running
  time is quadratic in $n$.
\end{proof}

In general, finitely generated subgroups and quotients of automatic
groups need not be automatic --- they need not even be finitely
presented. A subgroup $H$ of a finitely generated group $G=\langle
A\rangle$ is \emph{quasi-convex}\index{subgroup!quasi-convex} if there
exists a constant $\delta$
such that every $h\in H$ is connected to $\one\in G$ by a geodesic in
$\mathscr C(G,A)$ that remains at distance $\le\delta$ from
$H$. Typical examples are finite-index subgroups, free factors, and
direct factors.

On the other hand, a subgroup $H$ of an automatic group $G$ with
language $L(\mathcal L)$ is \emph{$\mathcal L$-rational} if the full
preimage of $H$ in $L(\mathcal L)$ is rational. The following is easy
but fundamental:
\begin{lemma}[\cite{MR1114609}]
  A subgroup $H$ of an automatic group is quasi-convex if and only if
  it is $\mathcal L$-rational.
\end{lemma}

It is still unknown whether automatic groups have solvable conjugacy
problem; however, there are asynchronously automatic groups with
unsolvable conjugacy problem, for instance appropriate amalgamated
products of two free groups over finitely generated subgroups. These
groups are asynchronously automatic~\cite{MR1147304}*{Theorem E},
and have unsolvable conjugacy problem~\cite{MR0310044}.
\begin{theorem}[\Gersten-\Short]
  Biautomatic groups have solvable conjugacy
  problem\index{biautomatic~group!conjugacy~problem}.
\end{theorem}
\begin{proof}[Sketch of proof; see~\cite{MR1117155}]
  Consider two words $x,y\in\At^*$. Using the biautomatic
  structure, the language
  \[C(x,y)=\{(u,v)\in\Ah^*\times\Ah^*\mid u,v\in\mathcal
  L\text{ and }\pi(u)=\pi(xvy)\}\]
  is rational.
  Now $x,y$ are conjugate if and only if $C(x^{-1},y)\cap\{(w,w)\mid
  w\in\mathcal L\}$ is non-empty. The problem of deciding whether a
  rational language is empty is algorithmically solvable.
\end{proof}
In fact, the centralizer of an element of a biautomatic group is a
quasi-convex subgroup, and is thus
biautomatic~\cite{MR1114609} (but we remark that it is still unknown
whether a quasi-convex subgroup of an automatic group is necessarily
automatic). There is therefore a good algorithmic 
description of \emph{all} elements that conjugate $x$ to $y$.

\subsection{Word-hyperbolic groups}\index{word-hyperbolic~group|(}
\Gromov\ \cite{MR919829} introduced the fundamental concept
of ``negative curvature'' to group theory. This goes further in the
direction of viewing groups as metric spaces, through the geodesic
distance on their Cayley graph. The definition is given for
\emph{geodesic} metric spaces, i.e., metric spaces in which any two
points can be joined by a geodesic segment:
\begin{definition}[\cites{MR1086648,MR1170363,MR1075994}]\label{BS2:def:hyp}
  Let $X$ be a geodesic metric space, and let $\delta>0$ be
  given. The space $X$ is \emph{$\delta$-hyperbolic} if, for any three
  points $A,B,C\in X$ and geodesics arcs $a,b,c$ joining them, every
  $P\in a$ is at distance at most $\delta$ from $b\cup c$.

  The space $X$ is \emph{hyperbolic}\index{hyperbolic!space} if it is
  $\delta$-hyperbolic for some $\delta$.  The finitely generated group
  $G=\langle A\rangle$ is
  \emph{word-hyperbolic}\index{word-hyperbolic~group}\index{group!word-hyperbolic}
  if it acts by isometries on a hyperbolic metric space $X$ with
  discrete orbits, finite point stabilizers, and compact quotient
  $X/G$.

  Equivalently, $G$ is word-hyperbolic if and only if $\mathscr
  C(G,A)$ is hyperbolic.
\end{definition}

\Gilman\ \cite{MR1985464} gave a purely automata-theoretic
definition of word-hyperbolic groups: $G$ is word-hyperbolic if and
only if, for some regular combing $\mathcal M\subset\At^*$, the
language $\{u\one v\one w\mid u,v,w\in\mathcal
M,\pi(uvw)=\one\}\subset\Ah^*$ is context-free. Using the geometric
definition, we note immediately the following examples: first, the
hyperbolic plane $\mathbb H^2$ is hyperbolic (with $\delta=\log3$); so
is $\mathbb H^n$. Any discrete, cocompact group of isometries of
$\mathbb H^n$ is word-hyperbolic. This applies in particular to the
surface group $J_g$ from~\eqref{BS2:eq:jg}, if $g\ge2$. Note however
that some word-hyperbolic groups are not small cancellation groups,
for instance because for small cancellation groups the complex in
Proposition~\ref{BS2:prop:class} has trivial homology in dimension
$\ge3$; yet the complex associated with a cocompact group acting on
$\mathbb H^n$ has infinite cyclic homology in degree $n$
(see~\cite{MR2365352} for applications of topology to group theory).

It is also possible to define $\delta$-hyperbolicity for spaces $X$
that are not geodesic (as, e.g., a group):
\begin{definition}
  Let $X$ be a metric space, and let $\delta'>0$ be given. The space
  $X$ is \emph{$\delta'$-hyperbolic} if, for any four points
  $A,B,C,D\in X$, the numbers
  \[\{d(A,B)+d(C,D),d(A,C)+d(B,D),d(A,D)+d(B,C)\}\] are such that the
  largest two differ by at most $\delta'$.
\end{definition}

Word-hyperbolic groups arise naturally in geometry, in the following
way: let $\mathcal M$ be a compact Riemannian manifold with negative
(not necessarily constant) sectional curvature. Then $\pi_1(\mathcal
M)$ is a word-hyperbolic group.

Word-hyperbolic groups are also ``generic'' among finitely-presented
groups, in the following sense: fix a number $k$ of generators, and a
constant $\epsilon\in[0,1]$. For large $N$, there are
$\approx(2k-1)^N$ words of length $\le N$ in $F_k$; choose a subset
$\mathscr R$ of size $\approx(2k-1)^{\epsilon N}$ of them uniformly at
random, and consider the group $G$ with presentation $\langle
A\mid\mathscr R\rangle$.

Then, with probability $\to1$ as $N\to\infty$, the group $G$ is
word-hyperbolic. Furthermore, if $\epsilon<\frac12$, then with probability
$\to1$ the group $G$ is infinite, while if $\epsilon>\frac12$, then
with probability $\to1$ the group $G$ is
trivial~\cite{MR1167524}.

Word-hyperbolic groups provide us with a large number of examples of
automatic groups. Better:
\begin{theorem}[\Gersten-\Short, \Gromov] Let $G$ be a word-hyperbolic
  group. Then $G$ is
  biautomatic\index{word-hyperbolic~group!is~biautomatic}. Moreover,
  the normal form $\mathcal L$ may be chosen to consist of geodesics.
\end{theorem}
\noindent Even better, the automatic structure is, in some precise sense,
unique~\cite{MR2151425}.

\begin{proof}[Sketch of proof]
  In a $\delta$-word-hyperbolic group $G$, geodesics
  $(2\delta+1)$-fellow-travel. On the other hand, $G$ has a finite
  number of cone types~\eqref{BS2:eq:cone}\index{cone~type}, so the
  language of geodesics is rational, recognized by an automaton with
  as many states as there are cone types.
\end{proof}

Hyperbolic spaces $X$ have a natural \emph{hyperbolic
  boundary}\index{hyperbolic!boundary}
$\partial X$: fix a point $x_0\in X$, and consider
\emph{quasi-geodesics at $x_0$}\index{quasi-geodesic}, namely
quasi-isometric embeddings
$\gamma:\N\to X$ starting at $x_0$. Declare two such quasi-geodesics
$\gamma,\delta$ to be equivalent if $d(\gamma(n),\delta(n))$ is
bounded. The set of equivalence classes, with its natural topology, is
the boundary $\partial X$ of $X$. The fundamental tool in studying
hyperbolic spaces is the following
\begin{lemma}[\Morse]
  Let $X$ be a hyperbolic space and let $C$ be a constant. There is
  then a constant $D$ such that all $C$-quasi-geodesics between two
  points $x,y\in X$ are at distance at most $D$ from one another.
\end{lemma}

The hyperbolic boundary $\partial X$ is compact, under appropriate
conditions satisfied e.g.\ by $X=\mathscr C(G,A)$, and $X\cup\partial
X$ is a compactification of $X$. Now, in that case, the automaton
$\mathcal L$ provides a symbolic coding of $\partial X$ as a finitely
presented shift space (where the shift action is the ``geodesic
flow'', following one step along infinite paths $\in\Ah^\infty$
representing quasi-geodesics).

We note that, for word-hyperbolic groups, the word and conjugacy
problem admit extremely efficient solutions: they are both solvable in
linear time by a Turing machine. The word problem is actually solvable
in real time, namely with a bounded amount of calculation allowed
between inputs~\cite{MR1758286}\index{word~problem}. The isomorphism
problem\index{isomorphism~problem} is decidable for word-hyperbolic
groups, say given by a finite presentation~\cite{1002.2590}
\nindex{Dahmani, Fran\c cois}\nindex{Guirardel, Vincent}.
Word-hyperbolic groups also satisfy a linear isoperimetric inequality,
in the sense that every $w\in F_A$ that evaluates to $\one$ in $G$ is
a product of $\mathcal O(|w|)$ conjugates of relators. Better:
\begin{proposition}
  A finitely presented group is word-hyperbolic if and only if it
  satisfies a linear isoperimetric
  inequality\index{isoperimetric~inequality}\index{word-hyperbolic~group!linear~isoperimetric~inequality},
  if and only if it
  satisfies a subquadratic isoperimetric inequality.
\end{proposition}

Note that the generalized word problem is known to be unsolvable
\cite{MR642423}, but the order problem is on the other hand solvable
in word-hyperbolic 
groups~\cite{MR1776048}.\index{generalized~word~problem}\index{order~problem}
\index{word-hyperbolic~group|)} It follows that the
generalized word problem is unsolvable for automatic groups as
well. \index{automatic~group|)} 

There are important weakenings of the definition of word-hyperbolic
groups; we mention two.  A \emph{bicombing}\index{bicombing} is a
choice, for every pair of vertices $g,h\in\mathscr C(G,A)$, of a path
$\ell_{g,h}$ from $g$ to $h$. Since $G$ acts by left-translation on
$\mathscr C(G,A)$, it also acts on bicombings. A bicombing satisfies
the \emph{$k$-fellow-traveller property} if for any neighbours $x'$ of
$x$ and $y'$ of $y$, the paths $\ell_{x,y}$ and $\ell_{x',y'}$
$k$-fellow-travel.

A \emph{semi-hyperbolic group}\index{group!semi-hyperbolic} is a group
admitting an invariant bicombing by fellow-travelling
words. See~\cite{MR1744486}, or the older paper~\cite{MR1300841}.  In
particular, biautomatic, and therefore word-hyperbolic, groups are
semi-hyperbolic.

Semi-hyperbolic groups are finitely presented and have solvable word
and conjugacy problems. In fact, they even have the ``monotone
conjugation property'', namely, if $g$ and $h$ are conjugate, then
there exists a word $w$ with $g^{\pi(w)}=h$ and
$|g^{\pi(w(i))}|\le\max\{|g|,|h|\}$ for all $i\in\{0,\dots,|w|\}$.

A group $G$ is \emph{relatively
  hyperbolic}\index{group!relatively~hyperbolic} \cite{MR1650094}
if it acts properly
discontinuously on a hyperbolic space $X$, preserving a family
$\mathcal H$ of separated horoballs, such that $(X\setminus\mathcal
H)/G$ is compact. All fundamental groups of finite-volume negatively
curved manifolds are relatively
hyperbolic.\index{fundamental~group!of~negatively~curved~manifold}

A non-closed manifold has ``cusps'', going off to infinity, whose
interpretation in the fundamental group are conjugacy classes of loops
that may be homotoped arbitrarily far into the
cusp. Farb~\cite{MR1650094} captures combinatorially the notion of
relative hyperbolicity as follows: let $\mathscr H$ be a family of
subgroups of a finitely generated group $G=\langle A\rangle$. Modify
the Cayley graph of $G$ as follows: for each coset $gH$ of a subgroup
$H\in\mathscr H$, add a vertex $gH$, and connect it by an edge to
every $gh\in\mathscr C(G,A)$, for all $h\in H$. In addition, require
that every edge in $\widehat{\mathscr C(G,A)}$ belong to only finitely
many simple loops of any given length.  The group $G$ is \emph{weakly
  relatively hyperbolic}, relative to the family $\mathscr H$, if that
modified Cayley graph $\widehat{\mathscr C(G,A)}$ is a hyperbolic
metric space.

By virtue of its geometric characterization, being word-hyperbolic is
a geometric property in the sense of Definition~\ref{BS2:def:geom}
(though beware that being hyperbolic is preserved under quasi-isometry
only if the metric spaces are geodesic). Being combable and being
bicombable are also geometric. 

We finally remark that a notion of word-hyperbolicity has been defined
for semigroups~\cites{MR2023798,MR2055042}; the definition uses
context-free languages. As for automatic (semi)groups, the theory does
not seem uniform enough to warrant a simultaneous treatment of groups
and semigroups; again, there is no clear geometric counterpart to the
definition of word-hyperbolic semigroups --- except in particular
cases, such as monoids defined through special confluent rewriting
systems~\cite{MR2536187}.
\index{word-hyperbolic~group|)}

\subsection{Non-automatic groups} All known examples of non-automatic
groups arise as groups violating some interesting consequence of
automaticity.

First, infinitely presented groups cannot be automatic. There are
uncountably many finitely generated groups, and only countably many
finitely presented groups; therefore automatic groups should be
thought of as the rationals among the real numbers.

Groups with unsolvable word problem cannot be automatic.

If a nilpotent group\index{group!nilpotent}\index{nilpotent~group} is
automatic, then it contains an abelian
subgroup of finite index~\cite{MR1703070}; therefore, for instance,
the discrete Heisenberg group\index{group!Heisenberg}\index{Heisenberg~group}
\begin{equation}\label{BS2:eq:heis}
  G=\begin{pmatrix}1&\Z&\Z\\0&1&\Z\\0&0&1\end{pmatrix}
  =\langle x,y\mid [x,[x,y]],[y,[x,y]]\rangle
\end{equation}
is not automatic. Note also that $G$ satisfies a cubic, but no
quadratic, isoperimetric inequality.

Many solvable groups have larger-than-quadratic isoperimetric
functions; they therefore cannot be automatic~\cite{MR1877218}. This
applies in particular to the Baumslag-Solitar
groups\index{group!Baumslag-Solitar}\index{Baumslag-Solitar~group}
\begin{equation}\label{BS2:eq:bs}
  BS_{1,n}=\langle a,t\mid a^n=a^t\rangle.
\end{equation}

Similarly, $\operatorname{SL}_n(\Z)$, for $n\ge 3$, or
$\operatorname{SL}_n(\mathcal O)$ for $n\ge2$, where $\mathcal O$ are
the integers in an imaginary number field, are not automatic.

Infinite, finitely generated torsion groups cannot be automatic: they
cannot admit a rational normal form, because of the pumping lemma. We
will see examples, due to \Grigorchuk\index{Grigorchuk~group} and
\Gupta-\Sidki\index{Gupta-Sidki~group}, in~\S\chapterref{BS2:sec:examples}.

There are combable groups that are not automatic~\cite{MR2016694}, for instance
\[G=\langle a_i,b_i,t_i,s\mid
t_1a_1=t_2a_2,[a_i,s]=[a_i,t_i]=[b_i,s]=[b_i,t_i]=\one\quad(i=1,2)\rangle,\]
which satisfies only a cubic isoperimetric inequality. Finitely presented
subgroups of automatic groups need not be automatic~\cite{MR1489138}.

The following group is asynchronously automatic, but is not automatic:
it does not satisfy a quadratic isoperimetric
inequality~\cite{MR1147304}*{\S11}:
\[G=\langle a,b,t,u\mid a^t=ab,b^t=a,a^u=ab,b^u=a\rangle.\]

\section{Groups generated by automata}\label{BS2:sec:2}\index{automata~group|(}
We now turn to another important class of groups related to
finite-state automata. These groups act by permutations on a set $A^*$
of words, and these permutations are represented by \emph{Mealy
  automata}\index{Mealy~automaton}.  These are deterministic, initial
finite-state transducers $\mealy$ with input and output alphabet $A$,
that are complete with respect to input; in other words,
\begin{equation}\label{BS2:eq:mealy}
  \text{At every state and for each $a\in A$,
    there is a unique outgoing edge with input $a$.}
\end{equation}

The automaton $\mealy$ defines a transformation of $A^*$, which
extends to a transformation of $A^\omega$, as follows. Given
$w=a_1a_2\cdots\in A^*\cup A^\omega$, there is by~\eqref{BS2:eq:mealy} a
unique path in $\mealy$ starting at the initial state and with
input labels $w$. The image of $w$ under the transformation is the
output label along that same path.
\begin{definition}
  A map $f:A^*\to A^*$ is
  \emph{automatic}\index{transformation!automatic}\index{mapping!automatic}
  if $f$ is produced by a
  finite-state automaton as above.
\end{definition}

One may forget the initial state of $\mealy$, and consider the set of
all transformations corresponding to all choices of initial state of
$\mealy$; the \emph{semigroup of the
  automaton}\index{semigroup!automaton@of~an~automaton} $S(\mealy)$ is
the semigroup generated by all these transformations. It is closely
connected to \Krohn-\Rhodes\ Theory~\cite{MR0175718}. Its relevance to
group theory was seen during \Glushkov's seminar on
automata~\cite{MR0138529}.

The automaton $\mealy$ is \emph{invertible} if furthermore it is
complete with respect to output; namely,
\begin{equation}\label{BS2:eq:invmealy}
  \text{At every state and for each $a\in A$,
    there is a unique outgoing edge with output $a$;}
\end{equation}
the corresponding transformation of $A^*\cup A^\omega$ is then
invertible; and the set of such permutations,
for all choices of initial state, generate the \emph{group of the
  automaton}\index{group!automaton@of~an~automaton} $G(\mealy)$. Note
that $S(\mealy)$ may be a proper subsemigroup of $G(\mealy)$, even if
$\mealy$ is \emph{invertible}.  General references on groups generated
by automata are~\cites{MR2162164,MR1841755,MR2035113}.

As our first, fundamental example, consider the automaton with
alphabet $A=\{0,1\}$
\begin{equation}\label{BS2:eq:adding}
  \mathcal T:\begin{fsa}[baseline]
    \node[state] (t)              {$t$};
    \node[state] (e) [right of=t] {$\one$};
    \path (t) edge[loop left] node {$1|0$} (t)
              edge node {$0|1$} (e)
          (e) edge[right,in=30,out=60,loop] node {$0|0$} (e)
              edge[right,in=300,out=330,loop] node {$1|1$} (e);
  \end{fsa}
\end{equation}
in which the input $i$ and output $o$ of an edge are represented as
`$i|o$'. The transformation associated with state $\one$ is clearly
the identity transformation, since any path starting from $\one$ is a
loop with same input and output. Consider now the transformation
$t$. One has, for instance, $t\cdot 111001=000101$, with the path
consisting of three loops at $t$, the edge to $\one$, and two loops at
$\one$. In particular, $G(\mathcal T)=\langle t\rangle$. We will see
in~\S\chapterref{BS2:sec:rev} that it is infinite cyclic.

\begin{lemma}\label{BS2:lem:prodautom}
  The product of two automatic transformations is automatic. The
  inverse of an invertible automatic transformation is automatic.
\end{lemma}
The proof becomes transparent once we introduce a good notation. If in
an automaton $\mealy$ there is a transition from state $q$ to
state $r$, with input $i$ and output $o$, we write
\begin{equation}\label{BS2:eq:biset}
  q\cdot i=o\cdot r.
\end{equation}
In effect, if the state set of $\mealy$ is $Q$, we are encoding $\mealy$
by a function $\tau:Q\times A\to A\times Q$. It then follows
from~\eqref{BS2:eq:mealy} that, given $q\in Q$ and $v=a_1\cdots a_n\in
A^*$, there are unique $w=b_1\cdots b_n\in A^*,r\in Q$ such that
$q\cdot a_1\cdots a_n=b_1\cdots b_n\cdot r$. The image of $v$ under the
transformation $q$ is $w$. We have in fact extended naturally the
function $\tau$ to a function $\tau:Q\times A^*\to A^*\times Q$.

\begin{proof}[Proof of Lemma~\ref{BS2:lem:prodautom}]
  Given $\mealy,\mathcal N$ initial automata with state sets $Q,R$
  respectively, consider the automaton $\mathcal{MN}$ with state set
  $Q\times R$ and transitions defined by $(q,r)\cdot i=q\cdot(r\cdot
  i)=o\cdot(q',r')$. If $q_0,r_0$ \ifnazi are\else be\fi\
  the initial states of $\mealy,\mathcal N$, then the transformation
  $q_0\circ r_0$ is the transformation corresponding to state
  $(q_0,r_0)$ in $\mathcal{MN}$.

  Similarly, if $q_0$ induces an invertible transformation, consider
  the automaton $\mealy^{-1}$ with state set $\{q^{-1}\mid q\in
  Q\}$ and transitions defined by $q^{-1}\cdot o=i\cdot r^{-1}$
  whenever~\eqref{BS2:eq:biset} holds. The transformation induced by
  $q_0^{-1}$ is the inverse of $q_0$.
\end{proof}
This construction applies naturally to any composition of finitely
many automatic transformations. In case they all arise from the same
machine $\mealy$, we \emph{de facto} extend the function $\tau$ describing 
$\mealy$ to a function $\tau:Q^*\times A^*\to A^*\times Q^*$,
and (if $\mealy$ is invertible) to a function $\tau:F_Q\times A^*\to
A^*\times F_Q$. It projects to a function $\tau:S(\mealy)\times A^*\to
A^*\times S(\mealy)$, and, if $\mealy$ is invertible, to a function
$\tau:G(\mealy)\times A^*\to A^*\times G(\mealy)$.

Note that a function $G(\mealy)\times A\to A\times G(\mealy)$
naturally gives a function, still written $\tau:G(\mealy)\to
G(\mealy)^A\rtimes\operatorname{Sym}(A)$; this is the semidirect
product of functions $A\to G(\mealy)$ by the symmetric group of $A$
(acting by permutation of \ifnazi coordinates\else co\"ordinates\fi),
and is commonly called the \emph{wreath product}\index{wreath~product}
$G(\mealy)\wr\operatorname{Sym}(A)$, see also Chapter~16.

This wreath product decomposition also inspires a convenient
description of the function $\tau$ by a
\emph{matrix embedding}\label{BS2:matrix}\index{matrix~embedding};
the size and shape of the matrix is determined by the
permutation of $A$, and the nonzero entries by the elements in
$G(\mealy)^A$; more precisely, assume $A=\{1,\dots,d\}$, and, for
$\tau(q)=((s_1,\dots,s_d),\pi) \in
G(\mealy)^A\rtimes\operatorname{Sym}(A)$, write $\tau'(q)=$ the
permutation 
matrix with $s_i$ at position $(i,i\pi)$.  Then these matrices
multiply as wreath product elements. More algebraically, we have
defined a homomomorphism $\tau':\Bbbk G\to M_d(\Bbbk G)$, where $\Bbbk
G$ is the group ring of $G$ over the field $\Bbbk$. Such an embedding
defines an algebra acting on the linear span of $A^*$; this algebra
has important properties, studied in~\cite{MR1423285} for
Gupta-Sidki's\index{Gupta-Sidki~group} example and in~\cite{MR2254535}
for Grigorchuk's\index{Grigorchuk~group} example.

The action of $g\in G(\mealy)$ may be described as follows: given a
sequence $u=a_1\cdots a_n$, compute $\tau(g,u)=(w,h)$. Then $g\cdot
u=w$; and the image of $g\cdot(uv)=w(h\cdot v)$; that is, the action
of $g$ on sequences starting by $u$ is defined by an element $h\in
G(\mealy)$ acting on the tail of the sequence. More geometrically, we
can picture $A^*$ as an infinite tree. The action of $g$ carries the
subtree $uA^*$ to $wA^*$, and, within $uA^*$ naturally identified with
$A^*$, acts by the element $h$. For that reason, $G(\mealy)$ is called
a \emph{self-similar group}\index{group!self-similar}.

The formalism expressing a Mealy machine as a map $\tau:Q\times A\to
A\times Q$ is completely symmetric with respect to $A$ and $Q$; the
\emph{dual}\index{Mealy~automaton!dual} of the automaton $\mealy$ is
the automaton $\mealy^\vee$
with state set $A$, alphabet $Q$, and
transitions given by $i\cdot q=r\cdot o$ whenever~\eqref{BS2:eq:biset}
holds. For example, the dual of~\eqref{BS2:eq:adding} is
\begin{equation}\label{BS2:eq:adding*}
  \mathcal T^\vee:\begin{fsa}[baseline]
    \node[state] (0)              {$0$};
    \node[state] (1) [right of=0] {$1$};
    \path (0) edge[loop left] node {$\one|\one$} (0)
              edge[bend left] node {$t|\one$} (1)
          (1) edge[bend left] node {$t|t$} (0)
              edge[loop right] node {$\one|\one$} (1);
  \end{fsa}
\end{equation}

In case the dual $\mealy^\vee$ of the automaton $\mealy$ is itself
an invertible automaton, $\mealy$ is called
\emph{reversible}\index{Mealy~automaton!reversible}\index{reversible~Mealy~automaton}. If
$\mealy$, $\mealy^\vee$ and $(\mealy^{-1})^\vee$ are all invertible, then
$\mealy$ is \emph{bireversible}\index{Mealy~automaton!bireversible}\index{bireversible~Mealy~automaton}; it
then has eight associated
automata, obtained through all combinations of $()^{-1}$ and $()^\vee$.

Note that $\mealy^\vee$ naturally encodes the action of $S(\mealy)$ on
$A$: it is a graph with vertex set $A$, and an edge, with (input)
label $q$, from $a$ to $q\cdot a$. More generally, $(\mealy^n)^\vee$
encodes the action of $S(\mealy)$ on the set $A^n$ of words of length
$n$.

More generally, we will consider subgroups of $G(\mealy)$, namely
subgroups generated by a subset of the states of an automaton; we call
these groups \emph{automata
  groups}\index{automata~group}\index{group!automata}. This is a large
class of groups, which contains in particular finitely generated
linear groups, see Theorem~\ref{BS2:thm:linear} below
or~\cite{MR1492064}. The elements of automata groups are, strictly
speaking, automatic permutations of $A^*$. It is often convenient to
identify them with a corresponding automaton, for instance constructed
as a power of the original Mealy automaton (keeping in mind the
construction for the composition of automatic transformations), with
appropriate initial state.

\begin{theorem}[\Brunner-\Sidki]\label{BS2:thm:linear}
  The affine group\index{group!affine} $\Z^n\rtimes\GL_n(\Z)$ is an automata
  group for each $n$.
\end{theorem}
This will be proven in more generality in~\S\chapterref{BS2:sec:rev}.

We mention some closure properties of automata groups. Clearly a
direct product of automata groups is an automata group (take the
direct product of the alphabets). A more subtle operation, called
\emph{tree-wreathing} in~\cites{MR1942271,MR2197828}, gives wreath
products with $\Z$.

A more general class of groups has also been considered, and is
relevant to~\S\chapterref{BS2:sec:img}: \emph{functionally recursive
  groups}\index{group!functionally~recursive}. Let $A$ denote a finite
alphabet, $Q$ a finite set, and
$F=F_Q$ the free group on $Q$. The ``automaton'' now is given by a set
of rules of the form
\[q\cdot a=b\cdot r\] for all $q\in Q,a\in A$, where $b\in A$ and
$r\in F$. In effect, in the dual $\mealy^\vee$ we are allowing
arbitrary words over $Q$ as output symbols.

\subsection{Main examples}\label{BS2:sec:examples}
Automata groups gained significance when simple examples of finitely
generated, infinite torsion groups, and groups of intermediate
word-growth, were discovered. \Aleshin\ \cite{MR0301107} studied
the
automaton\index{Al\"eshin~group}\index{group!Al\"eshin}~\eqref{BS2:eq:aleshin},
and showed that $\langle A,B\rangle$ is an infinite torsion
group. Another of his examples is described
in~\S\chapterref{BS2:sec:bireversible}.

\Grigorchuk~\cites{MR565099,MR784354,MR781246,MR764305,MR712546}
simplified Al\"eshin's example as follows: let $\mathcal A$ be
obtained from the Al\"eshin automaton by removing the gray states; the
state set of $\mathcal A$ is $\{\one,a,b,c,d\}$. He showed that
$G(\mathcal A)$, which is known as the \emph{Grigorchuk
  group},\index{Grigorchuk~group}\index{group!Grigorchuk} is an
infinite torsion group; see Theorem~\ref{BS2:thm:burnside}. In fact,
$G(\mathcal A)$ and $\langle A,B\rangle$ have isomorphic finite-index
subgroups.

\Gupta\ and \Sidki~\cites{MR759409,MR696534} constructed for all prime
$p$ an infinite, $p$-torsion group; their construction, for $p=3$, is
the automata group $G(\mathcal
G)$\index{Gupta-Sidki~group}\index{group!Gupta-Sidki}\index{group!p-@$p$-}\index{p-group@$p$-group}
generated by the automaton~\eqref{BS2:eq:GS}.

All invertible automata with at most three states and two alphabet
letters have been listed in~\cite{MR2432182}; here are some important
examples.

The 
affine group $BS_{1,3}=\{z\mapsto
3^pz+q/3^r\mid p,q,r\in\Z\}$\index{group!Baumslag-Solitar}\index{Baumslag-Solitar~group},
see~\eqref{BS2:eq:bs} is a linear group, and 
an automata
group by Theorem~\ref{BS2:thm:affine}; see also~\cite{MR2216703}. It
is generated by the automaton~\eqref{BS2:eq:BS13}.

As another important example, consider the lamplighter
group\index{group!lamplighter}\index{lamplighter~group}\index{wreath~product}
\begin{equation}
  G=(\Z/2)^{(\Z)}\rtimes\Z=\langle a,t\mid a^2,[a,a^{t^n}]\text{ for
    all }n\in\Z\rangle.
\end{equation}
It is an automata group~\cite{MR1866850}, embedded as the set of maps
\[\{z\mapsto (t+1)^pz+q\mid p\in\Z,q\in\mathbb F_2[t+1,(t+1)^{-1}]\}\]
in the affine group of $\mathbb F_2[[t]]$. It is generated by the
automaton $\mathcal L$ in~\eqref{BS2:eq:lamplighter}.

\begin{equation}\label{BS2:eq:aleshin}
  \begin{fsa}[baseline]
    \node at (-0.6,1.1) {$\mathcal A$};
    \node[state] (b) at (-4,2) {$b$};
    \node[state] (c) at (-1.6,2) {$c$};
    \node[state] (d) at (0.8,2) {$d$};
    \node[state] (a) at (-2.8,0) {$a$};
    \node[state] (e) at (0,0) {$\one$};
    \node[state,gray] (A) at (-3.7,-2.1) {$A$};
    \node[state,gray] (1d) at (1.1,-1.9) {};
    \node[state,gray] (B) at (-1.4,-1.8) {$B$};
    \path (a) edge[-implies,double] node {$0|1,1|0$} (e)
          (b) edge node {$1|1$} (c) edge node {$0|0$} (a)
          (c) edge node {$1|1$} (d) edge node[near start] {$0|0$} (a)
          (d) edge[bend right] node[above] {$1|1$} (b) edge node {$0|0$} (e)
          (e) edge[loop right] node {$i|i$} (e)
          (A) edge[bend left,gray] node {$0|0$} (b) edge[bend right,gray] node[below] {$1|1$} (1d)
          (B) edge[gray] node[near start] {$0|1$} (e) edge[gray] node {$1|0$} (a)
          (1d) edge[gray] node {$0|0$} (e) edge[bend right=40,gray] node[right] {$1|1$} (d);
  \end{fsa}
\end{equation}

\begin{equation}\label{BS2:eq:GS}
  \mathcal G:\begin{fsa}[baseline]
    \node[state] (t) at (-2.8,0) {$t$};
    \node[state] (T) at (2.8,0) {$t^{-1}$};
    \node[state] (a) at (0,1.6) {$a$};
    \node[state] (A) at (0,-1.6) {$a^{-1}$};
    \node[state] (e) at (0,0) {$\one$};
    \path (t) edge[bend left=10] node {$0|0$} (a)
              edge[bend right=10] node[below] {$1|1$} (A)
              edge[loop left] node {$2|2$} (t)
          (T) edge[bend left=10] node {$0|0$} (A)
              edge[bend right=10] node[above] {$1|1$} (a)
              edge[loop right] node {$2|2$} (T)
          (a) edge[-implies,double] node {$i|i+1$} (e)
          (A) edge[-implies,double] node[right] {$i|i-1$} (e)
          (e) edge[loop right] node {$i|i$} (e);
\end{fsa}
\end{equation}

\begin{equation}\label{BS2:eq:BS13}
  \begin{fsa}[baseline]
  \node[state,ellipse,inner sep=0mm] (x) {$3z$};
  \node[state,ellipse,inner sep=0mm] (y) [right of=x] {$3z+1$};
  \node[state,ellipse,inner sep=0mm] (z) [right of=y] {$3z+2$};
  \path (x) edge[loop left] node {$0|0$} (x) edge[bend left] node {$1|1$} (y)
  (y) edge[bend left] node {$0|1$} (x) edge[bend left] node {$1|0$} (z)
  (z) edge[loop right] node {$1|1$} (z) edge[bend left] node {$0|0$} (y);
\end{fsa}
\end{equation}

\begin{equation}\label{BS2:eq:lamplighter}
  \mathcal L:\begin{fsa}[baseline]
    \node[state,ellipse,inner sep=-1mm] (t)              {$(t+1)z$};
    \node[state,ellipse,inner sep=-1mm] (u) [right of=t] {$(t+1)z+1$};
    \path (t) edge[bend left] node {$1|0$} (u)
              edge[loop left] node {$0|0$} (t)
          (u) edge[loop right] node {$1|1$} (u)
              edge[bend left] node {$0|1$} (t);
  \end{fsa}
\end{equation}

The Basilica group\index{Basilica~group}\index{group!Basilica},
see~\cites{MR1902367,MR2176547}, will appear
again in~\S\chapterref{BS2:sec:img}. It is generated by the
automaton~\eqref{BS2:eq:bas}.

\begin{equation}\label{BS2:eq:bas}
  \mathcal B:\begin{fsa}[baseline]
    \node[state] (a) at (-2.7,1.3) {$a$};
    \node[state] (b) at (-2.7,-1.3) {$b$};
    \node[state] (e) at (0,0) {$\one$};
    \path (a) edge node {$0|1$} (e) edge[bend left] node {$1|0$} (b)
          (b) edge node[below] {$0|0$} (e) edge[bend left] node {$1|1$} (a)
          (e) edge[loop right] node {$0|0,1|1$} (e);
  \end{fsa}
\end{equation}

There are (unpublished) lists by \Sushchansky\ \emph{et al.} of all
(not necessarily invertible) automata with $\le3$ states, on a binary
alphabet; there are more than 2000 such automata; the invertible ones
are listed in~\cite{MR2432182}.

How about groups that are \emph{not} automata groups? Groups with
unsolvable word problem (or more generally whose word problem cannot
be solved in exponential time, see~\S\chapterref{BS2:sec:automprob}), and
groups that are not residually finite (or more generally that are not
residually (finite with composition factors of bounded order)) among
the simplest examples. In fact, it is difficult to come up with any
other ones.

\subsection{Decision problems}\label{BS2:sec:automprob}
One virtue of automata groups is that elements may easily be compared,
since (Mealy) automata admit a unique minimized form, which
furthermore may efficiently be computed in time $\mathcal O(\Card
A\Card Q\log\Card Q)$, see~\cites{MR0403320,MR1795249}.
\begin{proposition}
  Let $G$ be an automata group. Then the word problem is solvable in
  $G$\index{word~problem!in~automata~groups}\index{automata~group!word~problem},
  in at worst exponential time.
\end{proposition}
\begin{proof}
  Let $Q$ be a generating set for $G$, and for each $q\in Q$ compute
  the minimal automaton $\mealy_q$ representing $q$. Let $C$ be an
  upper bound for the number of states of any $\mealy_q$.

  Now given a word $w=q_1\cdots q_n\in(Q\sqcup Q^{-1})^*$, multiply the
  automata $\mealy_{q_1},\dots,\mealy_{q_n}$. The result is an
  automaton with $\le C^n$ states. Then $w$ is trivial if and only if
  all states to which the initial state leads have identical input and
  output symbols.
\end{proof}

It is unknown if the conjugacy or generalized word
problem\index{Grigorchuk~group!generalized~word~problem} are
solvable in general; though this is known in particular cases, such as
the Grigorchuk group $G(\mathcal A)$,
see~\cites{MR1687204,MR1687212,MR2009443}. The conjugacy problem is
solvable\index{Grigorchuk~group!conjugacy~problem} as soon as
$G(\mathcal A)$ is \emph{conjugacy separable}\index{group!conjugacy~separable},
namely, for $g,h$ non-conjugate in $G(\mathcal A)$ there exists a
finite quotient of $G(\mathcal A)$ in which their images are
non-conjugate. Indeed automata groups are recursively presented and
residually finite.

It is also unknown whether the order problem is solvable in arbitrary
automata groups; but this is known for particular cases, such as
bounded automata groups, see~\S\chapterref{BS2:sec:bounded}.

\Nekrashevych's limit space (see Theorem~\ref{BS2:thm:limit}) may
sometimes be used to prove that two contracting, self-similar groups
are non-isomorphic: By~\cite{MR2011117}, some groups admit essentially
only one weakly branch self-similar action\index{weakly~branch~group};
if the group is also contracting\index{contracting~automaton/group},
then its limit space is an isomorphism invariant.

On the other hand, in the more general class of
functionally recursive groups, the very solvability of the word
problem remains so far an open problem.

\subsection{Bounded and contracting automata}\label{BS2:sec:bounded}
As we saw in~\S\chapterref{BS2:sec:automprob}, it may be useful to note, and
use, additional properties of automata groups.
\begin{definition}
  An automaton $\mealy$ is \emph{bounded}\index{automaton!bounded} if
  the function which to
  $n\in\N$ associates the number of paths of length $n$ in $\mealy$
  that do not end at the identity state is a bounded function. A group
  is \emph{bounded}\index{group!bounded}\index{bounded~group}
if its
  elements are bounded automata; or,
  equivalently, if it is generated by bounded automata.
\end{definition}
More generally, \Sidki\ considered automata for which that
function is bounded by a polynomial; see~\cite{MR1774362}. He showed
in~\cite{MR2112674} that such groups cannot contain non-abelian free
subgroups.

\begin{definition}
  An automaton $\mealy$ is
  \emph{nuclear}\index{Mealy~automaton!nuclear}\index{automata~group!nuclear}
  if the set of
  recurrent
  states of $\mealy\mealy$ spans an automaton isomorphic to $\mealy$;
  and, for invertible $\mealy$, if additionally
  $\mealy=\mealy^{-1}$. Recall that a state is \emph{recurrent} if it
  is the endpoint of arbitrarily long paths.

  An invertible automaton $\mealy$ is
  \emph{contracting}\index{Mealy~automaton!contracting}\index{automata~group!contracting}\index{contracting~automaton/group}
  if $G(\mealy)=G(\mathcal N)$ for a (necessarily unique) nuclear
  automaton $\mathcal N$. The
  \emph{nucleus}\index{nucleus~of~a~Mealy~automaton} of $G(\mealy)$ is
  then $\mathcal N$.
\end{definition}
For example, the automata~(\ref{BS2:eq:aleshin},\ref{BS2:eq:GS})
are nuclear; the automata~(\ref{BS2:eq:adding},\ref{BS2:eq:bas}) are
contracting, with nucleus $\{1,t,t^{-1}\}$ and
$\{1,a^{\pm1},b^{\pm1},b^{-1}a,a^{-1}b\}$; the
automaton~\eqref{BS2:eq:lamplighter} is not contracting.

If $\mealy$ is contracting, then for every $g\in G(\mealy)$ there is a
constant $K$ such that (in the automaton describing $g$) all paths of
length $\ge K$ end at a state in $\mealy$. It also implies that there
are constants $L,m$ and $\lambda<1$ such that, for the word metric
$\|\cdot\|$ on $G(\mealy)$, whenever one has $g\cdot a_1\cdots
a_m=b_1\cdots b_m\cdot h$ with $h,g\in G(\mealy)$, one has
$\|h\|\le\lambda\|g\|+L$.

\begin{proposition}[\cite{MR2162164}*{Theorem~3.9.12}]
  Finitely generated bounded groups are
  contracting\index{bounded~group!is~contracting}.
\end{proposition}

Consider the following graph $\mathscr X(\mealy)$: its vertex set is
$A^*$. It has two kinds of edges, \emph{vertical} and
\emph{horizontal}. There is a vertical edge $(u,ua)$ for all $u\in
A^*,a\in A$, and a horizontal edge $(u,q\cdot u)$ for every $u\in
A^*,q\in Q$. Note that the horizontal and vertical edges form squares
labeled as in~\eqref{BS2:eq:biset}, and that the horizontal edges form
the Schreier graphs\index{graph!Schreier}\index{Schreier~graph} of the
action of $G(\mealy)$ on $A^n$.
\begin{proposition}[\cite{MR2162164}*{Theorem~3.8.6}]
  If $G(\mealy)$ is contracting then $\mathscr X(\mealy)$ is a
  hyperbolic graph\index{hyperbolic~graph} in the sense of
  Definition~\ref{BS2:def:hyp}.
\end{proposition}

Discrete groups may be broadly separated in two classes:
\emph{amenable} and \emph{non-amenable} groups. A group $G$ is
\emph{amenable}\index{amenable~group}\index{group!amenable} if it
admits a normalized, invariant mean, that is, a map $\mu:\mathcal
P(G)\to[0,1]$ with $\mu(A\sqcup B)=\mu(A)+\mu(B)$, $\mu(G)=1$ and
$\mu(gA)=\mu(A)$ for all $g\in G$ and $A,B\subseteq G$. All finite and
abelian groups are amenable; so are groups of subexponential
word-growth (see~\S\chapterref{BS2:sec:growth}). Extensions, quotients,
subgroups, and directed unions of amenable groups are amenable. On the
other hand, non-abelian free groups are non-amenable.

In understanding the frontier between amenable and non-amenable
groups, the Basilica group\index{Basilica~group} $G(\mathcal B)$
stands out as an important example: \Bartholdi\ and \Virag\ proved
that it is
amenable~\cite{MR2176547}\index{Basilica~group!is~amenable}, but its
amenability cannot be decided by the criteria of the previous
paragraph. We now briefly indicate the core of the argument.

The matrix embedding $\tau':\Bbbk G\to M_d(\Bbbk G)$ associated with a
self-similar group (see page~\pageref{BS2:matrix}) extends to a map
$\tau':\ell^1(G)\to M_d(\ell^1(G))$ on measures on $G$. A measure
$\mu$ gives rise to a random walk on $G$, with one-step transition
probability $p_1(x,y)=\mu(xy^{-1})$. On the other hand, $\tau'(\mu)$
naturally defines a random walk on $G\times X$; treating the second
variable as an ``internal degree of freedom'', one may sample the
random walk\index{random~walk} on $G\times X$ each time it hits
$G\times\{x_0\}$ for a fixed
$x_0\in X$. In favourable cases, the corresponding random walk on $G$
is \emph{self-similar}\index{random~walk!self-similar}: it is a convex
combination of $\one$ and
$\mu$. One may then deduce that its ``asymptotic entropy'' vanishes,
and therefore that $G$ is amenable. This strategy works in the
following cases:
\begin{theorem}[\Bartholdi-\Kaimanovich-\Nekrashevych~\cite{0802.2837}]
  Bounded groups are ame\-na\-ble.\index{bounded~group!is~amenable}
\end{theorem}
\begin{theorem}[\Amir, \Angel, \Virag\cite{0905.2007}]
  Automata of linear growth generate amenable groups.
\end{theorem}

\Nekrashevych\ conjectures that contracting automata
always generate amenable groups, and proves:
\begin{proposition}[\Nekrashevych,~\cite{0802.2554}]
  A contracting self-similar group cannot contain a non-abelian free subgroup.\index{free~group}
\end{proposition}

\noindent We turn to the original claim to fame of automata groups:
\begin{theorem}[\Aleshin-\Grigorchuk~\cites{MR0301107,MR565099}, \Gupta-\Sidki~\cite{MR759409}]
\index{group!p-@$p$-}\index{p-group@$p$-group}\index{Al\"eshin~group}\index{Grigorchuk~group!is~torsion}\index{Gupta-Sidki~group!is~torsion}
 \label{BS2:thm:burnside}
  The Grigorchuk group $G(\mathcal A)$ and the Gupta-Sidki group
  $G(\mathcal G)$ are infinite, finitely generated torsion groups.
\end{theorem}
\begin{proof}[Sketch of proof]
  To see that these groups $G$ are infinite, consider their action on
  $A^*$, the stabilizer $H$ of $0\in A\subset A^*$, and the
  restriction $\theta$ of the action of $H$ to $0A^*$. This defines a
  homomorphism
  $\theta:H\to\operatorname{Sym}(0A^*)\cong\operatorname{Sym}(A^*)$, which is
  in fact onto $G$. Therefore $G$ possesses a proper subgroup mapping
  onto $G$, so is infinite.

  To see that these groups are torsion, proceed by induction on the
  word-length of an element $g\in G$. The initial cases
  $a^2=b^2=c^2=d^2=\one$, respectively $a^3=t^3=\one$, are easily
  checked. Now consider again the action of $g$ on $A\subset A^*$. If
  $g$ fixes $A$, then its actions on the subsets $iA^*$ are again
  defined by elements of $G$, which are shorter by the contraction
  property; so have finite order. It follows that $g$ itself has
  finite order.

  If, on the other hand, $g$ does not fix $A$, then $g^{\Card A}$ fixes
  $A$; the action of $g^{\Card A}$ on $iA^*$ is defined by an element of
  $G$, of length at most the length of $g$; and (by an argument that
  we skip) smaller in the induction order than $g$; so $g^{\Card A}$ is
  torsion and so is $g$.
\end{proof}

Contracting groups have recursive presentations (meaning the relators
$\mathscr R$ of the presentation form a recursive subset of $F_Q$); in
favourable cases, such as branch
groups~\cite{MR2009317}\index{branch~group}\index{group!branch}, the
set of relators is the set of iterates, under an endomorphism of
$F_Q$, of a finite subset of $F_Q$. For example~\cite{MR819415},
Grigorchuk's\index{Grigorchuk~group!presentation} group satisfies
\[G(\mathcal A)=\langle a,b,c,d\mid
\sigma^n(bcd),\sigma^n(a^2),\sigma^n([d,d^a]),\sigma^n([d,d^{[a,c]a}])\text{ for
  all }n\in\N\rangle,\]
where $\sigma$ is the endomorphism of $F_{\{a,b,c,d\}}$
\begin{equation}\label{BS2:eq:sigma}
  \sigma:a\mapsto aca, b\mapsto d\mapsto c\mapsto b.
\end{equation}

\noindent A similar statement holds for the Basilica
group~\eqref{BS2:eq:bas}:\index{Basilica~group!presentation}
\[G(\mathcal B)=\langle a,b\mid
[a^p,(a^p)^{b^p}],[b^p,(b^p)^{a^{2p}}]\text{ for all }p=2^n\rangle;\]
here the endomorphism is $\sigma:a\mapsto b\mapsto a^2$.

\subsection{Branch groups}\index{weakly~branch~group|(}\index{branch~group|(}
Some of the most-studied examples of automata groups are \emph{branch
  groups}, see~\cite{MR1765119} or the survey~\cite{MR2035113}. We
will define a strictly smaller class:
\begin{definition}\label{BS2:def:branch}
  An automata group $G(\mealy)$ is \emph{regular weakly
    branch}\index{automata~group![regular]~[weakly]~branch}
  if it acts transitively on $A^n$ for all $n$, and if there exists a
  nontrivial subgroup $K$ of $G(\mealy)$ such that, for all $u\in A^*$
  and all $k\in K$, the permutation
  \[w\mapsto\begin{cases}u\,k(v) & \text{ if }w=uv,\\
    w & \text{ otherwise}\end{cases}\]
  belongs to $G(\mealy)$.

  The group $G(\mealy)$ is \emph{regular branch}\index{group!branch}
  if furthermore $K$ has finite index in $G(\mealy)$.
\end{definition}

If we view $A^*$ as an infinite tree, a regular branch group $G$
contains a rich supply of tree automorphisms in two manners: enough
automorphisms to permute any two vertices of the same depth; and, for
any disjoint subtrees of $A^*$, and for (up to finite index) any
elements of $G$ acting on these subtrees, an automorphism acting in
that manner on $A^*$.

In particular, if $G$ is a regular branch group, then $G$ and
$G\times\cdots\times G$, with $\Card A$ factors, have isomorphic
finite-index subgroups (they are \emph{commensurable},
see~\eqref{BS1:eq:commensurator}).

\begin{proposition}
  The  Grigorchuk group $G(\mathcal A)$ and the Gupta-Sidki group
  $G(\mathcal G)$ are regular branch\index{Grigorchuk~group!is~regular~branch}\index{Gupta-Sidki~group!is~regular~branch}; 
  the Basilica 
  group $G(\mathcal B)$ is regular weakly
  branch\index{Basilica~group!is~regular~weakly~branch}.
\end{proposition}
\begin{proof}[Sketch of proof]
  For $G=G(\mathcal A)$, note first that $G$ acts transitively on $A$;
  since the stabilizer of $0$ acts as $G$ on $0A^*$, by induction $G$
  acts transitively on $A^n$ for all $n\in\N$.

  Define then $x=[a,b]$ and $K=\llangle x\rrangle$. Consider the
  endomorphism~\eqref{BS2:eq:sigma}, and note that
  $\sigma(x)=[aca,d]=[x^{-1},d]\in K$ using the relation $(ad)^4=\one$, so
  $\sigma$ restricts to an endomorphism $K\to K$, such that
  $\sigma(k)$ acts as $k$ on $1A^*$ and fixes $0A^*$. Similarly,
  $\sigma^n(k)$ acts as $k$ on $1^nA^*$, so
  Definition~\ref{BS2:def:branch} is fulfilled for $u=1^n$. Since $G$
  acts transitively on $A^n$, the definition is also fulfilled for
  other $u\in A^n$.

  Finally, a direct computation shows that $K$ has index $16$ in $G$.

  The other groups $G(\mathcal G)$ and $G(\mathcal B)$ are handled
  similarly; for them, one takes $K=[G,G]$.
\end{proof}

Various consequences may be derived from a group being a branch
group; in particular,
\begin{theorem}[\Abert,~\cite{MR2143732}]
  A weakly branch group satisfies no
  identity\index{weakly~branch~group!satisfies~no~identity}; that is,
  if $G$ is a
  weakly branch group, then for every nontrivial word
  $w=w(x_1,\dots,x_k)\in F_k$, there are $a_1,\dots,a_k\in G$ such
  that $w(a_1,\dots,a_k)\neq\one$.
\end{theorem}
\index{branch~group|)}\index{weakly~branch~group|)}

\subsection{Growth of groups}\label{BS2:sec:growth}\index{group!growth|(}\index{growth~of~groups|(}
An important geometric invariant of a finitely generated group is the
asymptotic behaviour of its growth function $\gamma_{G,A}(n)$. Finite
groups, of course, have a bounded growth function. If $G$ has a
finite-index nilpotent subgroup, then $\gamma_{G,A}(n)$ is bounded by
a polynomial, and one says $G$ has \emph{polynomial
  growth}\index{growth~of~groups!polynomial}; the
converse is true~\cite{MR623534}.

On the other hand, if $G$ contains a free subgroup, for example if $G$
is word-hyperbolic and is not a finite extension of $\Z$, then
$\gamma_{G,A}$ is bounded from above and below by exponential
functions, and one says that $G$ has \emph{exponential
  growth}\index{growth~of~groups!exponential}.

By a result of \Milnor\ and \Wolf~\cites{MR0248688,MR0244899}, if $G$
has a solvable subgroup of finite index then $G$ has either polynomial
or exponential growth. The same conclusion holds, by \Tits'
alternative~\cite{MR0286898}, if $G$ is linear. \Milnor\
\cite{Milnor5603} asked whether there exist groups with growth strictly
between polynomial and exponential.

\begin{theorem}[\Grigorchuk~\cite{MR712546}]
  The Grigorchuk group $G(\mathcal A)$ has intermediate
  growth\index{growth~of~groups!intermediate}\index{Grigorchuk~group!has~intermediate~growth}. More
  precisely, its growth function satisfies the following estimates:
  \[e^{n^\alpha}\precsim\gamma_{G,S}(n)\precsim e^{n^\beta},\] with
  $\alpha=0.515$ and $\beta=\log(2)/\log(2/\eta)\approx0.767$, for
  $\eta\approx0.811$ the real root of the polynomial
  $X^3+X^2+X-2$.
\end{theorem}
\begin{proof}[Sketch of proof; see~\cites{MR1656258,MR1818662}]
  Recall that $G$ admits an endomorphism $\sigma$,
  see~\eqref{BS2:eq:sigma}, such that $\sigma(g)$ acts as $g$ on
  $1A^*$ and as an element of the finite dihedral group $D_8=\langle
  a,d\rangle$ on $0A^*$.

  Given $g_0,g_1\in G$ of length $\le N$, the element
  $g=a\sigma(g_0)a\sigma(g_1)$ has length $\le4N$, and acts (up to an
  element of $D_8$) as $g_i$ on $iA^*$ for $i=0,1$. It follows that
  $g$ essentially (i.e., up to $8$ choices) determines $g_0,g_1$,
  and therefore that $\gamma_{G,S}(4N)\ge(\gamma_{G,S}(N)/8)^2$. The
  lower bound follows easily.

  On the other hand, the Grigorchuk\index{Grigorchuk~group}
  group $G$ satisfies a stronger
  property than contraction; namely, for a well-chosen metric (which
  is equivalent to the word metric), one has that if $g\in G$ acts as
  $g_i\in G$ on $iA^*$, then
  \begin{equation}\label{BS2:eq:grigcontr}
    \|g_0\|+\|g_1\|\le\eta(\|g\|+1),
  \end{equation}
  with $\eta$ the constant above.

  Then, to every $g\in G$ one associates a description by a finite,
  labeled binary tree $\iota(g)$. If $\|g\|\le1/(1-\eta)$, its
  description is a one-vertex tree with $g$ at its unique
  leaf. Otherwise, let $i\in\{0,1\}$ be such that $ga^i$ fixes $A$,
  and write $g_0,g_1$ the elements of $G$ defined by the actions of
  $ga^i$ on $0A^*,1A^*$ respectively. Construct recursively the
  descriptions $\iota(g_0),\iota(g_1)$. Then the description of $g$ is
  a tree with $i$ at its root, and two descendants
  $\iota(g_0),\iota(g_1)$.

  By~\eqref{BS2:eq:grigcontr}, the tree $\iota(g)$ has at most
  $\|g\|^\beta$ leaves; and $\iota(g)$ determines $g$. There are
  exponentially many trees with a given number of leaves, and the upper
  bound follows.
\end{proof}

Among groups of exponential growth, \Gromov\ asked the following
question~\cite{MR682063}: is there a group $G$ of exponential growth,
namely such that $\lim\gamma_{G,Q}(n)^{1/n}>1$ for all (finite) $Q$, but such
that $\inf_{Q\subset G}\lim\gamma_{G,Q}(n)^{1/n}=1$?

Such examples, called \emph{groups of non-uniform exponential
  growth},\index{growth~of~groups!non-uniformly~exponential} were
first found by \Wilson~\cite{MR2031429}; see~\cite{MR1981466} for a
simplification. Both constructions are heavily based on groups
generated by automata.

It is known that essentially any function growing faster than $n^2$
may be, asymptotically, the growth function of a semigroup. It is
however notable that very small automata generate semigroups of growth
$\sim e^{\sqrt n}$, and of polynomial growth of irrational
degree~\cites{MR2394721,MR2194959}. However, it is not known whether
there exist groups whose growth function is strictly between
polynomial and $e^{\sqrt n}$.
\index{group!growth|)}\index{growth~of~groups|)}

\subsection{Dynamics and subdivision rules}\label{BS2:sec:img}
We show, in this subsection, how automata naturally arise from geometric or
topological situations. As a first step, we will obtain a functionally
recursive action; in favourable cases it will be encoded by an
automaton.  We must first adopt a slightly more abstract point of view
on functionally recursive groups:
\begin{definition}
  A group $G$ is \emph{self-similar}\index{group!self-similar} if it
  is endowed with a
  \emph{self-similarity biset}\index{self-similarity~biset}, that is,
  a set $\mathfrak B$ with
  commuting left and right actions, that is free qua right $G$-set.
\end{definition}
The fundamental example is $G=G(\mealy)$ and $\mathfrak B=A\times G$,
with actions
\[g\cdot(a,h)=(b,kh)\text{ if }\tau(g,a)=(b,k),\qquad (a,g)\cdot
h=(a,gh).\] Conversely, given a self-similar group $G$, choose a
\emph{basis} $A$ of its biset, i.e., express $\mathfrak B=A\times G$;
then define $\tau(g,a)=(b,k)$ whenever $g\cdot(a,1)=(b,k)$ in
$\mathfrak B$. This vindicates the notation~(\ref{BS2:eq:biset}).

Two bisets $\mathfrak B,\mathfrak B'$ are \emph{isomorphic} if there
is a map $\varphi:\mathfrak B\to\mathfrak B'$ with $g\varphi(b)h=\varphi(gbh)$
for all $g,h\in G,b\in\mathfrak B$. They are \emph{equivalent} if
there is a map $\varphi:\mathfrak B\to\mathfrak B'$ and an automorphism
$\theta:G\to G$ with $\theta(g)\varphi(b)\theta(h)=\varphi(gbh)$.

Consider now $X$ a topological space, and $f:X\to X$ a \emph{branched
  covering}\index{branched~covering}; this means that there is an open
dense subspace
$X_0\subseteq X$ such that $f:f^{-1}(X_0)\to X_0$ is a covering. The
subset $\mathscr C=X\setminus f^{-1}(X_0)$ is the \emph{branch locus},
and $\mathscr P=\bigcup_{n\ge1}f^n(\mathscr C)$ is the
\emph{post-critical locus}. Write $\Omega=X\setminus\mathscr P$, and
choose a basepoint $*\in\Omega$.

Two coverings $(f,\mathscr P_f)$ and $(g,\mathscr P_g)$ are
\emph{combinatorially
  equivalent}\index{branched~covering!combinatorially~equivalent} if
there exists a path $g_t$ through branched coverings, with
$g_0=f,g_1=g$, such that the post-critical set of $g_t$ varies
continuously along the path.

We define a self-similarity biset for $G=\pi_1(\Omega,*)$: set
\[\mathfrak B_f=\{\text{homotopy classes of paths }\gamma:[0,1]\to\Omega\mid\gamma(0)=f(\gamma(1))=*\}.\]
The right action of $G$ prepends a loop at $*$ to $\gamma$; the left
action appends the unique $f$-lift of the loop that starts at
$\gamma(1)$ to $\gamma$.

A choice of basis for $\mathfrak B$ amounts to choosing, for each
$x\in f^{-1}(*)$, a path $a_x\subset\Omega$ from $*$ to $x$. Set
$A=\{a_x\mid x\in f^{-1}(*)\}$. Now, for $g\in G$, and $a_x\in A$,
consider a path $\gamma$ starting at $x$ such that $f\circ\gamma=g$;
such a path is unique up to homotopy, by the covering property of
$f$. The path $\gamma$ ends at some $y\in f^{-1}(*)$. Set then
\[\tau(g,a_x)=(a_y,a_y^{-1}\gamma a_x),\]
where we write concatenation of paths in reverse order, that is,
$\gamma\delta$ is first $\delta$, then $\gamma$.

For example, consider the sphere $X=\widehat\C$, with branched
covering $f(z)=z^2-1$. Its post-critical locus is $\mathscr
P=\{0,-1,\infty\}$. A direct calculation (see e.g.~\cite{MR2091700})
gives that its biset is the automaton~\eqref{BS2:eq:bas}; the relevant
paths are shown here:
\[\begin{tikzpicture}[>=stealth]
\useasboundingbox (-6,-2) rectangle (6,2);
\node[gray] at (0,0) {\includegraphics[width=12cm,height=4cm]{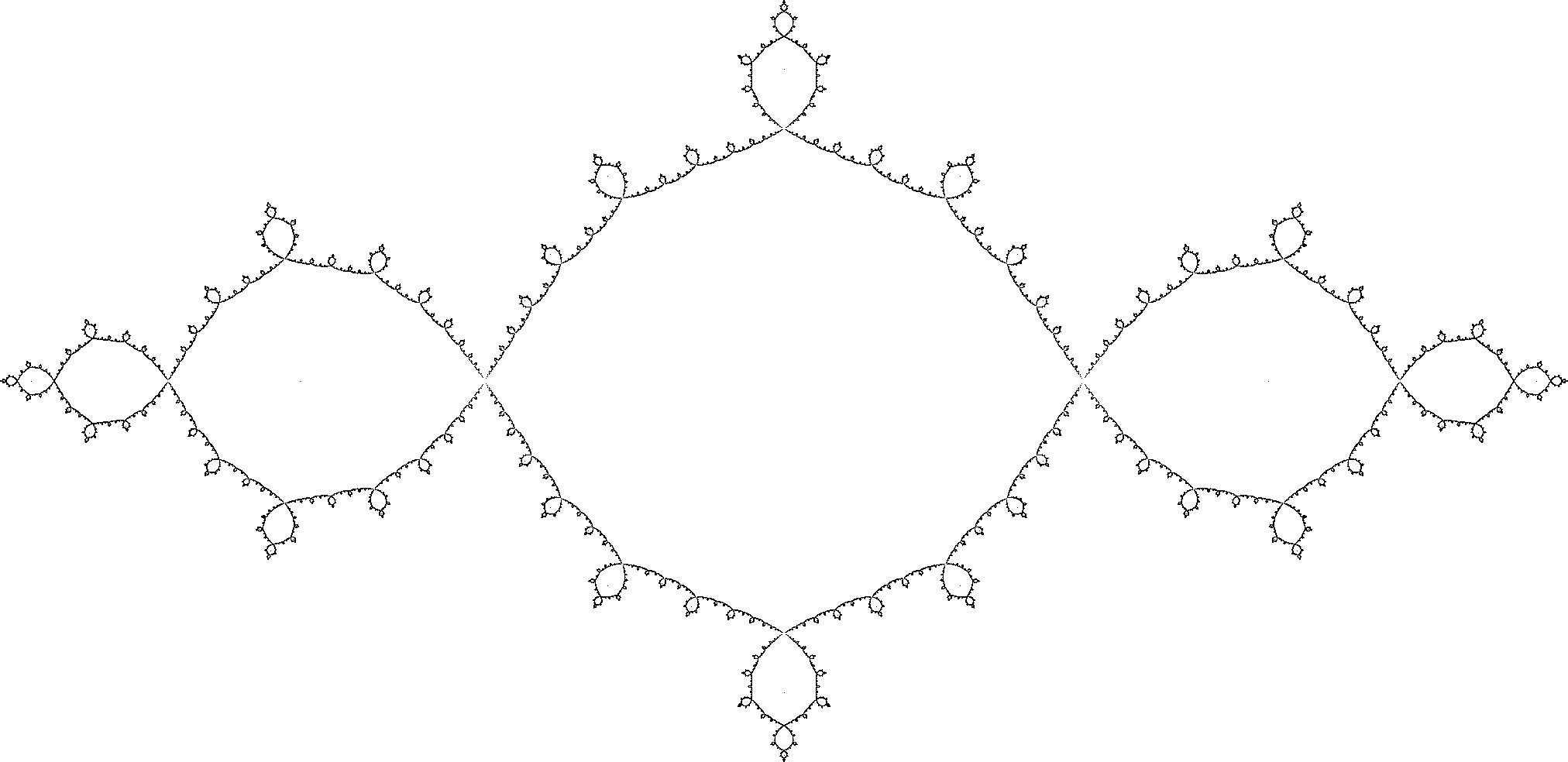}};
\node at (0,0) {$0$};
\node at (-3.6,0) {$-1$};
\node[inner sep=0pt] (*) at (-2.1,0) {$*$};
\node[inner sep=0pt] (x1) at (-2.6,0) {$x_1$};
\node[inner sep=0pt] (x0) at (2.6,0) {$x_0$};
\draw[thick,->] (*) .. controls (2.5,-4) and (2.5,4) .. node[left] {$a$} (*);
\draw[thick,->] (*) .. controls (-5,2.3) and (-5,-2.3) .. node [very near end,below] {$b$} (*);
\draw[thick,dashed,->] (x1) .. controls (-5.3,2) and (-5.3,-2) .. node [near start,above left=-1mm] {$f^{-1}(a)$} (x1);
\draw[thick,dashed,->] (x0) .. controls (5.3,-2) and (5.3,2) .. node [near start,below right=-1mm] {$f^{-1}(a)$} (x0);
\draw[thick,dashed,->] (x0) .. controls (1.4,2) and (-1.4,2) .. node [near start,above right=-1mm] {$f^{-1}(b)$} (x1);
\draw[thick,dashed,->] (x1) .. controls (-1.4,-2) and (1.4,-2) .. node [near start,below left=-1mm] {$f^{-1}(b)$} (x0);
\draw[thick,dotted,->] (*) -- (x1);
\draw[thick,dotted,->] (*) .. controls (-1.3,0.3) and  (1,1.6) .. node [below] {$a_{x_0}$} (x0);
\end{tikzpicture}\]

Branched self-coverings are encoded by self-similar groups in the
following sense:
\begin{theorem}[\Nekrashevych]
  Let $f,g$ be branched coverings. Then $f,g$ are combinatorially
  equivalent\Thurston[] if and only if the bisets $\mathfrak
  B_f,\mathfrak B_g$ are equivalent.
\end{theorem}
This result has been used to answer a long-standing open problem in
complex dynamics~\cite{MR2285317}.

If furthermore $G$ is finitely generated and the map $f$ expands a
length metric, then the associated biset may be defined by a
contracting automaton\index{contracting~automaton/group}. This is, in
particular, the case for all rational maps acting on the sphere
$\widehat\C$.

\begin{definition}
  Let $f:X\to X$ be a branched self-covering. The \emph{iterated
    monodromy group}\index{group!iterated~monodromy} of $f$ is the
  automata group $G(f)=G(\mealy)$,
  where $\mealy$ is an automaton describing the biset $\mathfrak B_f$.
\end{definition}

If $G=G(\mealy)$ is a contracting self-similar group, consider the
hyperbolic boundary $\mathscr J=\partial\mathscr X(\mealy)$, called
the \emph{limit space}\index{limit space} of $G$. It admits an
expanding self-covering map $s:\mathscr J\to\mathscr J$, induced on
vertices by the shift map $s(au)=u$.
\begin{theorem}[\cite{MR2162164}*{Theorems~5.2.6 and~5.4.3}]\label{BS2:thm:limit}
  The groups $G(s)$ and $G(\mealy)$ are isomorphic.

  Conversely, suppose $f$ is an analytic map, with \emph{Julia
    set}\index{Julia~set}
  $J$, the points near which $\{f^{\circ n}\mid n\in\N\}$ does not
  form a normal family. Then $(J,f)$ and $(\mathscr J,s)$ are
  homeomorphic and topologically conjugate.
\end{theorem}

For instance, the Julia set of the Basilica map $f(z)=z^2-1$ is
depicted above. Appropriately scaled and metrized, the Schreier graphs
of the action of $G(\mealy)$ on $X^n$ converge to $\mathscr J$.

The first appearance of encodings of branched coverings by automata
seems to be the ``finite subdivision rules'' by \Cannon, \Floyd\
and \Parry~\cite{MR1875951}; they wished to know when a branched
covering of the sphere may be realized as a conformal map. In their
work, a finite subdivision rule is given by a finite subdivision of
the sphere, a refinement of it, and a covering map from the refinement
to the original subdivision; by iteration, one obtains finer and finer
subdivisions of the sphere. The combinatorial information involved is
essentially equivalent to a self-similarity biset.  Contraction of
$G(\mealy)$ and combinatorial versions of expansion have been related
in~\cite{MR2520071}.

\subsection{Reversible actions}\label{BS2:sec:rev}
Recall that an automaton $\mealy$ is
\emph{reversible}\index{Mealy~automaton!reversible}\index{reversible~Mealy~automaton} if its dual
$\mealy^\vee$ is invertible. In other words, if $g\in G(\mealy)$, the
action of $g$ is determined by the action on any subset $uA^*$, for
$u\in A^*$.

We have already seen some examples of reversible automata,
notably~(\ref{BS2:eq:BS13},\ref{BS2:eq:lamplighter}). That last
example generalizes as follows: consider a finite group $G$, and set
$A=Q=G$. Define an automaton $\mathcal C_G$, the ``Cayley
automaton''\index{Mealy~automaton!Cayley~automaton}\index{Cayley~automaton}
of $G$, by $\tau(q,a)=(qa,qa)$. This automaton seems to have first
been considered in~\cite{MR0175718}*{page~358}\Krohn[]\Rhodes[]. The
automaton $\mathcal L$ in~\eqref{BS2:eq:lamplighter} is the special
case $G=\Z/2\Z$. The inverse of the automaton $\mathcal C_G$ is a
\emph{reset machine}\index{Mealy~automaton!reset~machine}, in that the
target of a transition depends only on the input, not on the source
state. \Silva\ and \Steinberg\ \cite{MR2197829} prove that, if $G$
is abelian, then $G(\mathcal C_G)=G\wr\Z$.

A large class of reversible automata is covered by the following
construction. Let $R$ be a ring, let $M$ be an $R$-module, and let $N$
be a submodule of $M$, with $M/N$ finite. Let $\varphi:N\to M$ be an
$R$-module homomorphism. Define a decreasing sequence of submodules
$M_i$ of $M$ by $M_0=M$ and $M_{n+1}=\varphi^{-1}(M_n)$, and denote by
$\End_R(M,\varphi)$ the algebra of $R$-endomorphisms of $M$ that map
$M_n$ into $M_n$ for all $n$. Assume finally that there is an algebra
homomorphism $\widehat\varphi:\End_R(M,\varphi)\to\End_R(M,\varphi)$ such that
$\varphi(an)=\widehat\varphi(a)\varphi(n)$ for all $a\in\End_R(M,\varphi),n\in
N$. Consider
\[T_M=\{z\mapsto az+m\mid a\in\End_R(M,\varphi),m\in M\}\]
the affine semigroup\index{group!affine} of $M$.

\begin{theorem}\label{BS2:thm:affine}
  Let $A$ be a transversal of $N$ in $M$. Then
  the semigroup $T_M$ acts self-similarly on $A^*$, by
  \[\tau(az+b,x)=(y,\widehat\varphi(a)z+\varphi(ax+b-y))\text{ for the
    unique $y\in A$ with }ax+b-y\in N.\]
  This action is
  \begin{enumerate}
  \item faithful if and only if $\bigcap_nM_n=0$;\label{BS2:enum:faithful}
  \item reversible if and only if $\varphi$ is injective;\label{BS2:enum:rev}
  \item defined by a finite-state automaton if $\widehat\varphi$ is an
    automorphism of finite order, and there exists a norm
    $\|\cdot\|:M\to\N$ such that $\|a+b\|\le\|a\|+\|b\|$, for all
    $K\in\N$ the ball $\{m\in M\mid K\ge\|m\|\}$ is finite, and a
    constant $\lambda<1$ satisfies $\|\varphi(n)\|\le\lambda\|n\|$ for
    all $n\in N$.\label{BS2:enum:autom}
  \end{enumerate}
\end{theorem}

We already saw some examples of this construction: the lamplighter
automaton $\mathcal L$ is obtained by taking $R=M=\mathbb F_2[t],
N=tM$, $\varphi(tm)=m$, $\widehat\varphi=1$, and $\|f\|=2^{\deg f}$ with
$\lambda=\frac12$. The semigroup $S(\mathcal L)$ is contained in
$T_M$, and the group $G(\mathcal L)$ is contained in the affine group
of $\mathbb F_2[[t]]$. More generally, the Cayley
automaton\index{Mealy~automaton!Cayley~automaton} of a finite group
$G$ is obtained by taking $R=G[[t]]$ with $G$ viewed as a ring with
product $xy=0$ unless $x=1$ or $y=1$.

The adding machine~\eqref{BS2:eq:adding} generates the subgroup of
translations in the affine group of $M$ with $R=M=\Z, N=2M$,
$\varphi(2m)=m$, and $\|m\|=|m|$. The same ring-theoretic data produce
the Baumslag-Solitar
group~\eqref{BS2:eq:BS13}\index{group!Baumslag-Solitar}; as above, we
use $R=\Z$ to obtain a semigroup, and $R=\Z_2$ (or any ring in which
$3$ is invertible) to obtain a group.

Consider, more generally, $R=\Z,M=\Z^n, N=2M$, and $\varphi(2m)=m$. These
data produce the affine group\index{group!affine} $\Z^n\rtimes\GL_n(\Z)$,
proving Theorem~\ref{BS2:thm:linear}.

A finer construction, giving an action on the binary tree, is to take
again $M=\Z^n$ and $N=\varphi^{-1}(M)$ with
$\varphi^{-1}(x_1,\dots,x_n)=(2x_n,x_1,\dots,x_{n-1})$; here
$\widehat\varphi(a)=\varphi\circ a\circ\varphi^{-1}$. This gives a faithful
action, on the binary tree, of
\[\Z^n\rtimes\{a\in\GL_n(\Z)\mid a\bmod 2\text{ is lower triangular}\}.\]

\begin{proof}[Sketch of proof]
  \eqref{BS2:enum:faithful} The action is faithful if and only if the
  translation part $\{z\mapsto z+m\}$ acts faithfully; and $z\mapsto
  z+m$ acts trivially on $A^*$ if and only if $m\in M_n$ for all
  $n\in\N$.

  \eqref{BS2:enum:rev} For any $x\in A$, the map (not a homomorphism!)
  $T_M\to T_M$ which to $g\in T_M$ associates the permutation of $A^*$
  given by $A^*\to xA^*\overset g\to g(x)A^*\to A^*$ is injective precisely when
  $\varphi$ is injective.

  \eqref{BS2:enum:autom} Without loss of generality, suppose
  $\widehat\varphi=1$. Consider $g=z\mapsto az+m\in T_M$. Let $K$ be
  larger than the norms of $ax+y$ for all $x,y\in A$. Then the states
  of an automaton describing $g$ are all of the form $z\mapsto az+m'$,
  with $\|m'\|\le(\|m\|+K)/(1-\lambda)$; there are finitely many
  possibilities for such $m'$.
\end{proof}

Note that the transversal $A$ amounts to a choice of ``digits'': the
analogy is clear in the case of the adding
machine~\eqref{BS2:eq:adding}, which has digits $\{0,1\}$ and
``counts'' in base $2$. For more general radix representations and
their association with automata, see~e.g.~\cite{MR1267355}.

\subsection{Bireversible actions}\label{BS2:sec:bireversible}
\index{Mealy~automaton!bireversible|(}\index{bireversible~Mealy~automaton|(}
Recall that an automaton $\mealy$ is bireversible if
$\mealy,\mealy^\vee,(\mealy^{-1})^\vee,((\mealy^\vee)^{-1})^\vee$
etc.\ are all invertible; equivalently, the map $\tau:Q\times A\to
A\times Q$ is a bijection for $Q$ the state set of
$\mealy\sqcup\mealy^{-1}$.

Bireversible automata are interpreted in~\cite{MR1841119} in terms of
\emph{commensurators} of free groups, defined
in~\eqref{BS1:eq:commensurator} of Chapter~\ref{chapterBS1}. Consider
a free group $F_A$ on a set $A$. Its Cayley graph $\mathscr C$ is a
tree, and $F_A$ acts by isometries on $\mathscr C$, so we have
$F_A\le\Isom(\mathscr C)$. Furthermore, $\mathscr C$ is oriented: its
edges are labeled by $A\sqcup A^{-1}$, and we choose as orientation the
edges labeled $A$. In this way, $F_A$ is contained in the
orientation-preserving subgroup of $\Isom(\mathscr C)$, denoted
$\overrightarrow{\Isom(\mathscr C)}$.
\begin{proposition}
  The stabilizer of $\one$ in
  $\operatorname{Comm}_{\overrightarrow{\Isom(\mathscr C)}}(F_A)$ is
  the set of bireversible automata with alphabet $A$.
\end{proposition}
\begin{proof}[Sketch of proof]
  The proof relies on an interpretation of finite-index subgroups of
  $F_A$ as complete automata, see~\S\ref{BS1:stallingsautomaton}.

  Let $\mealy$ be a bireversible automaton with alphabet $A$. Erase
  first the output labels from $\mealy$; this defines the \Stallings\
  automaton of a finite-index subgroup $H_1$ (of index $\Card Q$) of
  $F_A$. Erase then the input labels from $\mealy$; this defines an
  isomorphic subgroup $H_2$ of $F_A$. The automaton $\mealy$ itself
  defines an isomorphism between these two subgroups, which preserves
  the Cayley graph.

  Conversely, given an isometry $g$ of the Cayley graph of $F_A$ which
  restricts to an isomorphism $G\to H$ between finite-index subgroups
  of $F_A$, the \Stallings\ graphs of $G$ and $H$ and put their labels
  together, as input and output, to construct a bireversible
  automaton.
\end{proof}

It is striking that all known bireversible automata generate finitely
presented groups. There are, up to isomorphism, precisely two
minimized bireversible automata with three states and two alphabet
letters:
\[\begin{fsa}[baseline]
  \node[state] (a) at (-2.7,1.3) {$a$};
  \node[state] (c) at (-2.7,-1.3) {$c$};
  \node[state] (b) at (0,0) {$b$};
  \node at (-0.2,-1.3) {$\mathcal E_1$};
  \path (a) edge node {$0|0$} (b) edge[bend left] node {$1|1$} (c)
        (c) edge[-implies,double,bend left] node {$\begin{matrix}0|1\\1|0\end{matrix}$} (a)
        (b) edge node {$0|0$} (c) edge[loop right] node {$1|1$} (b);
  \end{fsa}\quad
  \begin{fsa}[baseline]
    \node[state] (a) at (-2.7,1.3) {$a$};
    \node[state] (c) at (-2.7,-1.3) {$c$};
    \node[state] (b) at (0,0) {$b$};
    \node at (-0.2,-1.3) {$\mathcal F_1$};
    \path (a) edge node {$0|1$} (b) edge[bend left] node {$1|0$} (c)
          (c) edge[-implies,double,bend left] node {$\begin{matrix}0|0\\1|1\end{matrix}$} (a)
          (b) edge node {$0|1$} (c) edge[loop right] node {$1|0$} (b);
\end{fsa}.\]
These automata are part of families, whose general term $\mathcal
E_n,\mathcal F_n$ has $2n+1$ states. We describe only $\mathcal F_n$:
\[\mathcal F_n:\begin{fsa}[baseline]
    \node[state] (a) at (2.5,1.5) {$a$};
    \node[state] (b) at (0,0) {$b$};
    \node[state] (c) at (2.5,-1.5) {$c$};
    \node[state] (d) at (5,-1.5) {$d$};
    \node[state] (m) at (7.5,-1.5) {$m$};
    \node[state] (n) at (7.5,1.5) {$n$};
    \node[state] (z) at (5,1.5) {$z$};
    \path (a) edge node[above] {$0|1$} (b) edge[bend left] node {$1|0$} (c)
          (b) edge node[below] {$0|1$} (c) edge[loop left] node {$1|0$} (b)
          (c) edge[-implies,double,color=gray] node[below] {$0|0,1|1$} (d)
          (d) edge[-implies,double,color=gray,dash pattern=on 1pt off 4pt] node[below] {$0|0,1|1$} (m)
          (m) edge[-implies,double,color=gray] node[right] {$0|0,1|1$} (n)
          (n) edge[-implies,double,color=gray,dash pattern=on 1pt off 4pt] node[above] {$0|0,1|1$} (z)
          (z) edge[-implies,double,color=gray] node[above] {$0|0,1|1$} (a);
\end{fsa}\]

\Aleshin\index{Al\"eshin~group!is~free}~\cite{MR713968} proved that
the group generated by the states $b_1,b_2$ in $\mathcal F_1,\mathcal
F_2$ respectively is a free group on its two generators; but his
argument (especially Lemma~8) has been considered incomplete, and a
detailed proof appears in~\cite{SVV}. \Aleshin's idea is to prove by
induction that, for any reduced word
$w\in\{b_1^{\pm1},b_2^{\pm1}\}^*$, the syntactic monoid of the
corresponding automaton acts transitively on its state set.

\Sidki\ conjectured that in fact $G(\mathcal F_1)$ is a free group on
its three generators; this has been proven in~\cite{MR2318547}. On the
other hand, $G(\mathcal E_1)$ is a free product of three cyclic groups of
order 2. Both
proofs illustrate some techniques used to compute with bireversible
automata. They rely on the following
\begin{lemma}\label{BS2:lem:birev}
  Let $L\subset Q^*$ be a subset mapping to $G(\mealy)$ through the
  evaluation map. If $L$ is $G(\mealy^\vee)$-invariant, and every
  $G(\mealy^\vee)$-orbit contains a word mapping to a nontrivial
  element of $G(\mealy)$, then $L$ maps injectively onto $G(\mealy)$.
\end{lemma}
To derive the structure of a bireversible group, we therefore seek
a $G(\mealy^\vee)$-invariant subset $L\subset Q^*$ that maps onto
$G(\mealy)\setminus\{1\}$, and show that every $G(\mealy)$-orbit
contains a non-trivial element of $G(\mealy)$.

\begin{theorem}[\Muntyan-\Savchuk]\index{free~product}
  $G(\mathcal E_1)=\langle a,b,c\mid a^2,b^2,c^2\rangle$.
\end{theorem}
\noindent Note that this result generalizes: $G(\mathcal E_n)$ is a
free product of $2n+1$ order-two groups.
\begin{proof} Write $Q=\{a,b,c\}$. We first check the relations
  $a^2=b^2=c^2=\one$ in $G=G(\mathcal E_1)$.  Let $L\subset Q^*$
  denote those sequences $s_1\cdots s_n$ with $s_i\neq s_{i+1}$ for all
  $i$.

  Consider the group $G(\mathcal E_1^\vee)$, with generators $0,1$. It
  acts on $L$, and acts transitively on $L\cap Q^n$ for all $n$;
  indeed already $0$ acts transitively on $Q=L\cap Q^1$, and $1$ acts
  on $\{a,c\}Q^{n-1}\cap L$ as a $2^n$-cycle, conjugate to the
  action~\eqref{BS2:eq:adding} in the sense that there is an
  identification of $\{a,c\}Q^{n-1}\cap L$ with $\{0,1\}^n$
  interleaving these actions. It follows that the $3\cdot2^{n-1}$
  elements of $L\cap A^n$ are in the same orbit.
  \[\mathcal E_1^\vee:\begin{fsa}[baseline]
    \node[state] (0) at (0,0) {$0$};
    \node[state] (1) at (3,0) {$1$};
    \path (0) edge[loop left] node {$\begin{matrix}a|b\\ b|c\end{matrix}$} ()
              edge[bend left] node {$c|a$} (1)
          (1) edge[loop right] node {$\begin{matrix}a|c\\ b|b\end{matrix}$} ()
              edge[bend left] node {$c|a$} (0);
  \end{fsa}\]

  It remains to note that $L\cap A^n$ contains a word mapping to a
  nontrivial element of $G$; for example, $c(ab)^{(n-1)/2}$ or
  $c(ab)^{n/2-1}a$ depending on the parity of $n$; and to apply
  Lemma~\ref{BS2:lem:birev}.
\end{proof}

\begin{theorem}[\Vorobets]\index{free~group}
  $G(\mathcal F_2)=\langle a,b,c\mid\emptyset \rangle\cong F_3$.
\end{theorem}
\noindent Note that this result generalizes: $G(\mathcal F_n)$ is a
free group of rank $2n+1$.
\begin{proof}[Sketch of proof]
  Again the key is to control the orbits of $G^\vee=G(\mathcal
  F_2^\vee)=\langle 0,1\rangle$ on the reduced words over
  $Q=\{a,b,c\}$ of any given length. Let $s\in(\pm1)^n$ be a sequence
  of signs, and consider
  \[L_s=\{w=w_1^{s_1}\cdots w_n^{s_n}\in(Q\sqcup Q^{-1})^*\mid
  w_i^{s_i}\neq w_{i+1}^{-s_{i+1}}\text{ for all }i\}.\] We show that
  $G^\vee$ acts transitively on $L_s$ for all $s$, and that $L_s$
  contains a word mapping to a nontrivial element of $G$. Consider the elements
  \[\alpha=0^21^{-2}0^21^{-1},\quad\beta=1^20^{-2}1^20^{-1},\quad\gamma=1^{-1}0,\quad\delta=01^{-1}\]
  of $G^\vee$, where the products are computed left-to-right; they are
  described by the automata
  \[\begin{fsa}
    \node[state] (a) at (0,0) {$\alpha$};
    \node[state] (b) at (4,0) {$\beta$};
    \node[state] (c) at (7,-1) {$\gamma$};
    \node[state] (d) at (9,-1) {$\delta$};
    \path (a) edge[loop left] node [below=2mm] {$\bm c^\pm|\bm c^\pm$} ()
              edge[bend left=20] node [near start,above=10] {$\bm a|\bm a$}
              node [near start,above] {$\bm b|\bm b$}
              node [near end,above=10] {$a^{-1}|b^{-1}$}
              node [near end,above] {$b^{-1}|a^{-1}$} (b);
    \path (b) edge[loop right] node [below=2mm] {$\bm c^\pm|\bm c^\pm$} ()
              edge[bend left=20] node [near start,below] {$a|b$}
              node [near start,below=10] {$b|a$}
              node [near end,below] {$\bm{a^{-1}}|\bm{a^{-1}}$}
              node [near end,below=10] {$\bm{b^{-1}}|\bm{b^{-1}}$} (a);
    \path (c) edge[loop above] node[above=20] {$\bm a^\pm|\bm a^\pm$}
              node[above=10] {$b^\pm|c^\pm$} node[above] {$c^\pm|b^\pm$} ();
    \path (d) edge[loop above] node[above=20] {$a^\pm|b^\pm$}
              node[above=10] {$b^\pm|a^\pm$}
              node[above] {$\bm c^\pm|\bm c^\pm$} ();
  \end{fsa}\] The elements $\gamma,\delta$ generate a copy of
  $\operatorname{Sym}(3)$, allowing arbitrary permutations of $Q$ or
  $Q^{-1}$. In particular, $G^\vee$ acts transitively on $L_s$
  whenever $|s|\le1$, so we may proceed by induction on $|s|$. The
  elements $\alpha,\beta$, on the other hand, fix a large set of
  sequences (following the bold edges in the automata).

  Consider now $s=s_1\cdots s_n$, and $s'=s_1\cdots s_{n-1}$. If
  $s_{n-1}\neq s_n$, so that $\Card{L_s}=2\Card{L_{s'}}$, then there
  exists $w=w_1^{s_1}\cdots w_n^{s_n}\in L_s$, moved by $\alpha$ or
  $\beta$, and such that $w_1^{s_1}\cdots w_{n-1}^{s_{n-1}}\in L_{s'}$
  is fixed by $\alpha$ and $\beta$; so $G^\vee$ acts transitively on
  $L_s$.

  If $s_1\neq s_2$, apply the same argument to $L_{s_n^{-1}\cdots
    s_1^{-1}}$ and $L_{s_n^{-1}\cdots s_2^{-1}}$.

  Finally, if $s_1=s_2$ and $s_{n-1}=s_n$, consider a typical $w\in
  L_{s_2\cdots s_{n-1}}$, and all $w_{qr}=q^{s_1}wr^{s_n}$, for $q,r\in
  Q$. Using the action of $\alpha$ and $\beta$, the words $w_{qa}$ and
  $w_{qb}$ are in the same $G^\vee$-orbit for all $q\in Q$, and similarly
  $w_{ar}$ and $w_{br}$ are in the same $G^\vee$-orbit for all $r\in
  Q$. For all $r\in Q$, finally, $w_{ar},w_{br'},w_{cr''}$ are in the
  same $G^\vee$-orbit for some $r',r''\in Q$, and similarly
  $w_{qa,q'b,q''c}$ are in the same $G^\vee$-orbit. It follows that all
  $w_{qr}$ are in the same $G^\vee$-orbit, so by induction $L_s$ is a
  single orbit.

  It remains to check that every $L_s$ contains a word $w$ mapping to
  a nontrivial group element. If $n$ is odd, set $w_i=a$ if $s_i=1$
  and $w_i=b$ if $s_i=-1$; then $\overline w$ acts nontrivially on
  $A$. If $n$ is even, change $w_n$ to $c^{s_n}$; again $\overline w$
  acts nontrivially on $A$. We are done by Lemma~\ref{BS2:lem:birev}.
\end{proof}

Burger\nindex{Burger,~Marc} and
Mozes\nindex{Mozes,~Shahar}~\cites{MR1446574,MR1839488,MR1839489} have
constructed some infinite,
finitely presented simple groups, see also~\cite{MR2320973}.  From
this chapter's point of view, these groups are obtained as follows:
one constructs an ``appropriate'' bireversible automaton $\mealy$ with
state set $Q$ and alphabet $A$, defines
\[G_0=\langle A\cup Q\mid aq=rb\text{ whenever that relation holds in
}\mealy\rangle,\]\index{VH-group}\index{group!VH-}
and considers $G$ a finite-index subgroup of $G_0$. We will not
explicitly give here the conditions required on $\mealy$ for their
construction to work; but note that automata groups can be understood
as a byproduct of their work.  Wise\nindex{Wise,~Daniel~T.}
constructed finitely presented groups with non-residual finiteness
properties that are also related to automata~\cite{MR2341837}.

Burger and Mozes give the following algebraic construction: consider two primes
$p,\ell\equiv1\pmod4$. Let $A$ (respectively $Q$) denote those integral
quaternions, up to a unit $\pm1,\pm i,\pm j,\pm k$, of norm $p$
(respectively $\ell$). By a result of Hurwitz\nindex{Hurwitz,~Adolf},
$\Card A=p+1$ and $\Card Q=\ell+1$.
Furthermore~\cite{H19}, for every $q\in Q,a\in A$ there are unique
(again up to units) $b\in A,r\in Q$ with $qa=br$. Use these relations
to define an automaton $\mealy_{p,\ell}$. Clearly $\mealy_{p,\ell}$ is
bireversible, with dual $\mealy_{p,\ell}^\vee=\mealy_{\ell,p}$. Again thanks to
unique factorization of integral quaternions of odd norm,
\begin{proposition}
  $G(\mealy_{p,\ell})=F_{(\ell+1)/2}$.
\end{proposition}

Glasner\nindex{Glasner,~Yair} and Mozes\nindex{Mozes,~Shahar}~\cite{MR2155175}
constructed an example of a bireversible automata
group with
Kazhdan's\index{Kazhdan~group}\index{group!Kazhdan}\index{property!t@(T)}
property (T).
\index{Mealy~automaton!bireversible|)}\index{bireversible~Mealy~automaton|)}
\index{automata~group|)}


\newpage
\begin{abstract}

Finite automata have been used effectively in recent years to define
infinite groups. The two main lines of research have as their most
representative objects the class of automatic groups (including
word-hyperbolic groups as a particular case) and automata groups
(singled out among the more general self-similar groups).

The first approach implements
in the language of automata some tight constraints on the geometry of
the group's Cayley graph, building strange, beautiful bridges between
far-off domains.  Automata are used to define a normal form
for group elements, and to monitor the fundamental group operations.

The second approach features groups acting in a finitely constrained
manner on a regular rooted tree. Automata define sequential
permutations of the tree, and represent the group elements
themselves. The choice of particular classes of automata has often provided
groups with exotic behaviour which have revolutioned our perception of
infinite finitely generated groups.
\end{abstract}

\index{group!automatic|see{automatic~group}}
\index{group!word-hyperbolic2@---|see{word-hyperbolic~group}}
\index{group!Gupta-Sidki|see{Gupta-Sidki~group}}
\index{group!Grigorchuk|see{Grigorchuk~group}}
\index{group!automata|see{automata~group}}
\index{group!Al\"eshin|see{Al\"eshin~group}}
\index{group!Basilica|see{Basilica~group}}
\index{group!growth|see{growth~of~groups}}
\index{hyperbolic!group|see{word-hyperbolic~group}}
\index{group!biautomatic|see{biautomatic~group}}
\index{group!bounded|see{bounded~group}}
\index{problem,~decision!word|see{word~problem}}
\index{problem,~decision!word!generalized2@---|see{generalized~word~problem}}
\index{problem,~decision!order|see{order~problem}}
\index{problem,~decision!membership2@---|see{membership~problem}}
\index{group!fundamental|see{fundamental~group}}

\printindex
\end{document}